\documentclass[11pt]{article}
\usepackage{amsmath,amssymb,amsthm}
\usepackage{algorithm}
\usepackage{algorithmic}
\usepackage{subcaption}
\usepackage{caption}
\usepackage[table]{xcolor}
\usepackage{xcolor}
\usepackage{tcolorbox}
\usepackage{tikz}
\usepackage{bbm}
\usepackage{wasysym}
\usepackage{geometry}
\usepackage{verbatim}
\usepackage{dirtytalk}
\usepackage{graphicx}
\usepackage{tikz}
\usepackage{amsfonts}
\usepackage{hyperref}
\usepackage{mathtools}
\usepackage{csquotes}
\usepackage{listings}
\usepackage{thmtools}
\usepackage{cleveref}
\usepackage{tikz-cd}
\usepackage{minibox}
\usepackage{braket}
\usepackage{soul}

\makeatletter
\newcommand\ketbra[2][]{%
  \def\ketbra@firstarg{#1}%
  \def\ketbra@secondarg{#2}%
  \ifx\ketbra@firstarg\empty%
    \left\lvert\ketbra@secondarg\middle\rangle \! \middle\langle \ketbra@secondarg\right\rvert%
  \else%
    \left\lvert\ketbra@firstarg\middle\rangle \! \middle\langle\ketbra@secondarg\right\rvert%
  \fi%
}
\makeatother

\makeatletter
\def\metadef#1#2{%
  \def\metadef@iter##1{\ifx##1;\else \expandafter\newcommand\csname#1\endcsname{#2}\expandafter\metadef@iter\fi}%
  \expandafter\metadef@iter%
}
\makeatother

\metadef{vec#1}{\mathbf #1}abcdefghijklmnopqrstuvwxyz;

\metadef{ten#1}{\mathbf #1}ABCDEFGHIJKLMNOPQRSTUVWXYZ;

\metadef{mat#1}{\mathbf #1}ABCDEFGHIJKLMNOPQRSTUVWXYZ;

\metadef{frak#1}{\mathfrak #1}ABCDEFGHIJKLMNOPQRSTUVWXYZ;

\metadef{c#1}{\mathcal #1}ABCDEFGHIJKLMNOPQRSTUVWXYZ;

\metadef{s#1}{\mathsf #1}ABCDEFGHIJKLMNOPQRSTUVWXYZ;

\newtheorem{theorem}{Theorem}
\newtheorem{lemma}{Lemma}
\newtheorem{corollary}{Corollary}
\newtheorem{proposition}{Proposition}
\newtheorem{remark}{Remark}

\newtheorem{definition}{Definition}

\newcommand{\enote}[1]{\noindent \textcolor{orange}{(\textbf{Eli:} #1\noindent)}}
\newcommand{\mnote}[1]{\noindent \textcolor{violet}{(\textbf{Matthew:} #1\noindent)}}
\renewcommand{\enote}[1]{}
\renewcommand{\mnote}[1]{}

\DeclareMathOperator{\Gen}{Gen}

\newcommand{\poly}{\mathsf{poly}}
\newcommand{\negl}{\mathsf{negl}}

\newcommand{\U}{\mathcal{U}}
\newcommand{\E}{\mathop{\mathbb{E}}}

\newcommand{\wt}[1]{\widetilde{#1}}

\newcommand{\abs}[1]{\left|#1\right|}
\renewcommand{\S}{\mathcal{S}}

\DeclareMathOperator{\argmax}{argmax}
\newcommand{\w}{\omega}

\DeclarePairedDelimiter\ceil{\lceil}{\rceil}

\newcommand{\N}{\mathbb{N}}

\makeatletter
\newcommand\ExpOwsg[4][]{%
  \def\ExpOwsg@firstarg{#1}%
  \ifx\ExpOwsg@firstarg\empty%
    \mathsf{Exp}_{#2,#3}(#4)
  \else%
    \mathsf{Exp}_{#2,#3,#1}(#4)
  \fi%
}
\makeatother

\newcommand{\QCMA}{\mathsf{QCMA}}

\newcommand{\PromBQP}{\mathsf{PromiseBQP}}
\newcommand{\PromQCMA}{\mathsf{PromiseQCMA}}
\newcommand{\PromQMA}{\mathsf{PromiseQMA}}

\newcommand{\xor}{\oplus}
\newcommand{\A}{\mathcal{A}}
\newcommand{\V}{\mathcal{V}}

\newcommand{\Samp}{\mathsf{Samp}}
\newcommand{\PuzzSamp}{\mathsf{PuzzSamp}}
\newcommand{\Ver}{\mathsf{Ver}}
\newcommand{\OWP}{\mathsf{OWPuzz}}
\newcommand{\EV}{\mathsf{EV}}
\newcommand{\PRS}{\mathsf{PRS}}
\newcommand{\OWSG}{\mathsf{OWSG}}
\newcommand{\EVOWP}{\mathsf{EV-OWPuzz}}
\newcommand{\QPRG}{\mathsf{QPRG}}
\newcommand{\Ext}{\mathsf{Ext}}

\usepackage{graphicx} 
\usepackage{authblk}
\definecolor{corlinks}{RGB}{200,0,0}
\definecolor{cormenu}{RGB}{200,0,0}
\definecolor{corurl}{RGB}{200,0,0}

\hypersetup{
 colorlinks=true,
 urlcolor=corlinks,
 linkcolor=corlinks,
 menucolor=cormenu,
 citecolor=corlinks,
 pdfborder= 0 0 0
}

\title{On Central Primitives for Quantum Cryptography with Classical Communication}
\author[1]{Kai-Min Chung}
\author[2]{Eli Goldin}
\author[3]{Matthew Gray}
\affil[1]{Academia Sinica (kmchung@iis.sinica.edu.tw)}
\affil[2]{New York University (eli.goldin@nyu.edu)}
\affil[3]{University of Oxford (matthew.gray@cs.ox.ac.uk)}

\begin{document}

\maketitle

\setcounter{tocdepth}{3}

\begin{abstract}
Recent work has introduced the ``Quantum-Computation 
Classical-Communication'' (QCCC) (Chung et. al.) setting for cryptography. 
There has been some evidence that One Way Puzzles ($\OWP$) are the natural central cryptographic primitive for this setting 
(Khurana and Tomer). For a primitive to be considered central it should have several characteristics. It should be well behaved 
(which for this paper we will think of as having amplification, combiners, and universal constructions); it should be implied by a 
wide variety of other primitives; and it should be equivalent to some class of useful primitives. 
We present combiners, correctness and security amplification, and a universal construction for $\OWP$. 
Our proof of security amplification uses a new and cleaner construction of EFI from $\OWP$ 
(in comparison to the result of Khurana and Tomer) that generalizes to weak $\OWP$ and is the most technically involved section of the paper. 
It was previously known that $\OWP$ are implied by other primitives of interest including commitments, symmetric key encryption, 
one way state generators ($\OWSG$), and therefore pseudorandom states ($\PRS$). However we are able to rule out $\OWP$'s equivalence to 
many of these primitives by showing a black box separation between general $\OWP$ and a restricted class of $\OWP$ 
(those with efficient verification, which we call $\EVOWP$). We then show that $\EVOWP$ are also implied by most of these primitives, 
which separates them from $\OWP$ as well. This separation also separates 
extending $\PRS$ from highly compressing $\PRS$ answering an open question of Ananth et. al.

\end{abstract}


{\let\clearpage\relax \tableofcontents}
{\let\clearpage\relax}
\thispagestyle{empty}


\section{Introduction}

In the realm of cryptography, there is perhaps no primitive more important than one-way functions. A one-way function is an efficiently computable deterministic function which is easy to compute, but hard to invert. Although at first glance the definition seems simple, one-way functions are special for several reasons. First and foremost, one-way functions are ``minimal." If modern cryptography exists in any form, then one-way functions must also exist~\cite{HILL99,IL89,Impag95}. Furthermore, pretty much all of these constructions are obvious. Second, one-way functions are ``useful." There is a large class of cryptographic primitives (known as Minicrypt) which can all be built from and are equivalent to one-way functions~\cite{Impag95}. Included in Minicrypt are symmetric key encryption, pseudorandom generators, and commitment schemes~\cite{HILL99,GGM86,Naor91}. Finally, one-way functions are ``well-behaved." They satisfy several natural properties~\cite{Levin87}, and are equivalent to most of their variants~\cite{Yao82,IL89}. Due to these three characteristics of one-way functions, one of the most useful things to do when trying to understand a new classical cryptographic primitive is to compare it to a one-way function.

This centrality of one-way functions no longer holds once quantum computation enters the picture. In particular, in the quantum setting, it seems that one-way functions are no longer minimal~\cite{Kretschmer21Quantum}. In particular, there exists a quantum oracle relative to which one-way functions do not exist, but quantum cryptography (in the form of pseudorandom state generators, quantum bit commitments, and many other primitives) is still possible. Recently, there has been strong evidence in support of a new simple primitive, the EFI pair, being minimal~\cite{khurana2024commitments,BCQ22}. An EFI pair is a pair of efficiently samplable quantum mixed states which are indistinguishable yet statistically far. Furthermore, EFI pairs are also useful. They can be used to build a large number of quantum cryptographic primitives, from quantum bit commitments to secure multiparty computation~\cite{BCQ22,AQY22}. Finally, EFI pairs are fairly well-behaved. The security of EFI pairs can be amplified~\cite{bostanci2023efficient}, there exists combiners and universal constructions for EFI pairs~\cite{HKNY23}, and EFI pairs are also equivalent to some of their variants~\cite{HMY23}. 

In the classical setting, it appears that one-way functions serve as an effective minimal primitive. In the quantum output setting, EFI pairs are a promising candidate for our minimal primitive. A number of recent works have also considered a hybrid setting, primitives where the cryptographic algorithms are quantum, but all communication and outputs are classical~\cite{ACCFLM22,ananth2023pseudorandom} ~\cite{CLM23,khurana2024commitments}. In the style of~\cite{ACCFLM22}, we will refer to this as the quantum computation classical communication (QCCC) setting. An immediate and natural question about this setting is ``what is a good central primitive?''

Just like in the fully quantum setting, it is unlikely that one-way functions can be a minimal primitive in the QCCC setting. In particular, there is a barrier to building one-way functions from a QCCC primitive known as the one-way puzzle~\cite{khurana2024commitments,Kretschmer21Quantum}. A one-way puzzle consists of an efficient quantum sampler which produces keys and puzzles along with a (possibly inefficient) verification procedure. The one-wayness corresponds to the idea that given a puzzle, it should be hard to find a matching key. Although one-way functions cannot serve as a central QCCC primitive, at first glance one-way puzzles make a fairly good candidate. In particular, one-way puzzles are minimal in the sense that almost all QCCC primitives can be used to build one-way puzzles~\cite{khurana2024commitments}. 

On the other hand, their well-behavedness and usefulness are less clear. It is known that one-way puzzles can be used to build EFI pairs (and thus everything which follows from EFI pairs)~\cite{khurana2024commitments}. However, as far as the authors are aware, there are no existing constructions of QCCC style primitives from one-way puzzles. The well-behavedness of one-way puzzles is similarly unstudied.

\paragraph*{Our results} In this work, we seek to investigate what primitives can be built from one-way puzzles, as well what useful properties one-way puzzles may or may not satisfy. Whether or not one-way puzzles are adopted as a central primitive in the same manner as one-way functions or EFI pairs is a community matter, but we hope that our results help shed light onto the question. To summarize our results, we show that
\begin{enumerate}
    \item There exists a robust combiner for one-way puzzles. That is, given two candidate one-way puzzles, there is a way to combine the candidates to get a construction which is secure as long as one of the candidates is secure.
    \item There exists a universal construction of a one-way puzzle. That is, a construction which is secure as long as one-way puzzles exist.
    \item There exist amplification theorems for one-way puzzles. That is, there is a method to take a one-way puzzle with weakened correctness or security guarantees and transform it into a full one-way puzzle.
    \item We show that one-way puzzles can be built from EFID pairs (the QCCC version of EFI pairs).
    \item We show that one-way puzzles are equivalent to one-way puzzles whose key is generated uniformly at random, answering an open question of~\cite{khurana2024commitments}.
    \item We show that one-way puzzles are equivalent to "distributional" one-way puzzles. For a distributional one-way puzzle, given a puzzle it is hard to sample from the distribution over keys.
\end{enumerate}

We also consider in detail an important restricted variant of one-way puzzles, which was first introduced under the name ``hard quantum planted problem for QCMA,"~\cite{KNY23}, but which we will refer to as efficiently verifiable one-way puzzles. We show the following results about this variant
\begin{enumerate}
    \item There exists combiners, a universal construction, and amplification theorems for efficiently verifiable one-way puzzles.
    \item Most QCCC primitives which can be used to build one-way puzzles can also be used to build efficiently verifiable one-way puzzles, with the notable exception of interactive commitment schemes. In particular, we show explicitly that pseudodeterministic PRGs and non-interactive commitments imply efficiently verifiable one-way puzzles.
    \item There exists a quantum oracle relative to which one-way puzzles and pseudorandom states exist but efficiently verifiable one-way puzzles do not.
\end{enumerate}

The last two points here together provide a barrier to building most QCCC primitives from one-way puzzles. Perhaps this means that efficiently verifiable one-way puzzles make a better candidate for centrality. However, if QCCC commitments can be built from one-way puzzles, then it may make sense to treat one-way puzzles as a central primitive on a lower level than efficiently verifiable one-way puzzles. We compare the relationship to the separation between one-way functions and one-way permutations.

In addition to this, since pesudorandom states also exist under the same oracle, our results provide a barrier to building most QCCC primitives from pseudorandom states. 

Note that our separation in fact separates efficiently verifiable one-way puzzles from pseudorandom state generators with linear output. But it is not hard to show that pseudorandom state generators with logarithmic output can be used to build efficiently verifiable one-way puzzles. Thus, our final separation also provides a barrier to length reduction for pseudorandom state generators, answering an open question of~\cite{ALY23}.

A summary of known relationships between QCCC primitives is included in~\Cref{fig:graph}.

\paragraph*{A better construction of EFI pairs from one-way puzzles}
The most technically demanding of our results is, surprisingly, the amplification theorem for one-way puzzles. It turns out that due to the inefficient nature of verification, most natural techniques fail. The techniques we use to achieve amplification for one-way puzzles can be also be used to construct EFI pairs from one-way puzzles, recreating a result from~\cite{khurana2024commitments}. In addition, our construction has several advantages over the existing construction in the literature.

First, the proof of security for our construction is significantly more straightforward than the existing argument. In particular, the argument does not rely on techniques dealing explicitly with a preimage space (such as leftover hash lemma or Goldreich-Levin), and so more naturally fits with the quantum nature of the primitive. Second, our construction produces an EFI pair even when instantiated with a one-way puzzle with weakened security guarantees. This is the essential reason that this technique is useful for proving amplification. 

\usetikzlibrary{positioning} 
\usetikzlibrary{calc} 
\usetikzlibrary {quotes}
\tikzset{>=latex} 

\tikzstyle{mysmallarrow}=[->,black,line width=1.6]
\tikzstyle{mydbarrow}=[<->,black,line width=1.6]
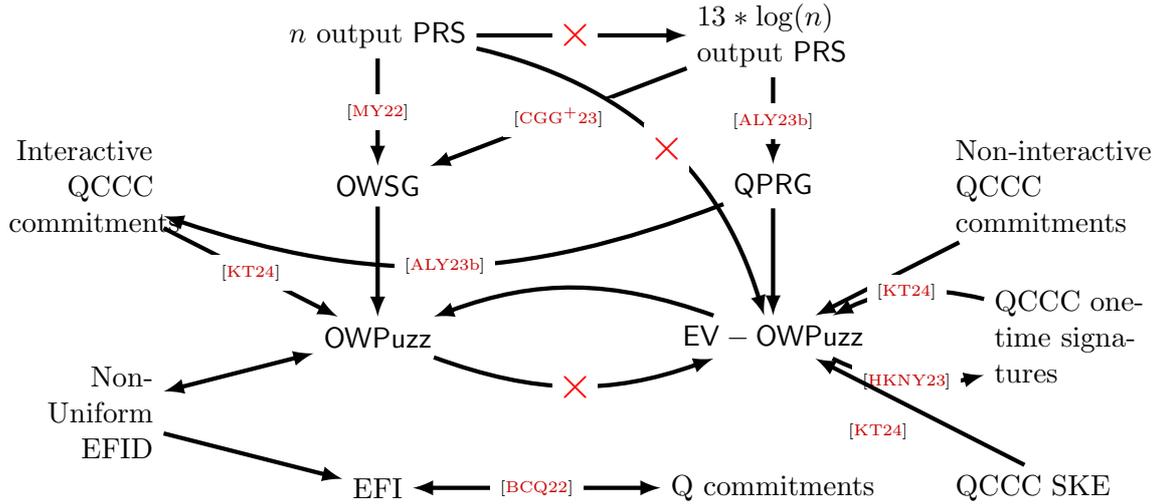
\begin{figure}
\begin{center}
    \begin{tikzpicture}[scale=1,every edge quotes/.style = {font=\footnotesize,fill=white}]
      \def\h{-2.0} 
      \def\w{2.6} 

        \node[] (n) at (-\w,0) {%
        $n$ output $\PRS$};
        
        \node[text width=2cm] (log) at (\w,0) {%
        $13*\log(n)$ output $\PRS$};

        \node[] (owsg) at (-\w,\h) {%
        $\OWSG$};
        
        \node[] (qprg) at (\w,\h) {%
        $\QPRG$};

        \node[text width=3cm] (qc-com)  at (2.5*\w,1*\h) {
        Non-interactive \\ QCCC \\ commitments};

        \node[text width=1.9cm, align=right] (int-com) at (-2.5*\w,\h) {%
        Interactive \\ QCCC \\ commitments};
        
        \node[] (owp) at (-\w,2*\h) {%
        $\OWP$};

        \node[] (evowp) at (\w,2*\h) {%
        $\EVOWP$};

        \node[text width=2.5cm] (qc-sig)  at (2.6*\w,2*\h) {
        QCCC one-time signatures};

        \node (efid)[text width=1.9cm, align=right]  at (-2.5*\w,2.5*\h) {
        Non-Uniform EFID};
        

        \node (efi)  at (-\w,3*\h) {
        EFI};

        \node (q-com)  at (\w,3*\h) {
        Q commitments};

        \node (qc-ske)[text width=3cm]  at (2.5*\w,3*\h) {
        QCCC SKE};
        
        \draw[mysmallarrow=black] (n) edge[left,->] (log);
        \node[fill=white] (X1) at (0,0) {$\textcolor{red}{\bigtimes}$};
        \draw[mysmallarrow=black](log) edge["\tiny{\cite{ALY23}}"] (qprg);
        \draw[mysmallarrow=black](log) edge["\tiny{\cite{cavalar2023computational}}"]  (owsg);
        \draw[mysmallarrow=black](n) edge["\tiny{\cite{morimae2022quantum}}"] (owsg);
    \draw[mysmallarrow=black][bend left=30] (n) edge[left,->] (evowp);
    \node[fill=white] (X1) at (1.2,-1.5) {$\textcolor{red}{\bigtimes}$};
        \draw[mysmallarrow=black][bend left=20] (qprg) edge["\tiny{\cite{ALY23}}"] (int-com);
        \draw[mysmallarrow=black] (owsg) edge (owp);
        \draw[mysmallarrow=black] (qprg) edge (evowp);
        \draw[mysmallarrow=black] (qc-com) edge (evowp);
        \draw[mysmallarrow=black] (int-com) edge["\tiny{\cite{khurana2024commitments}}",opacity=0.8] (owp);
        \draw[mysmallarrow=black][bend right=20] (owp) edge (evowp);
        \draw[mysmallarrow=black][bend right=20] (evowp) edge (owp);
        \node[fill=white] (X2) at (0,2.33*\h) {%
        $\textcolor{red}{\bigtimes}$};

        \draw[mysmallarrow=black][bend right=20] (evowp) edge["\tiny{\cite{HKNY23}}",fill opacity=0.8] (qc-sig);
        \draw[mysmallarrow=black][bend right=20] (qc-sig) edge["\tiny{\cite{khurana2024commitments}}",fill opacity=0.8] (evowp);
        
        \draw[mydbarrow] (efid) edge (owp);
        \draw[mysmallarrow=black] (efid) edge (efi);
        \draw[mysmallarrow=black] (qc-ske) edge[style={auto=left},"\tiny{\cite{khurana2024commitments}}"] (evowp);

    \draw[mydbarrow] (efi) edge["\tiny{\cite{BCQ22}}"] (q-com);
      
    \end{tikzpicture}
\end{center}
\caption{All implications known about one-way puzzles ($\OWP$) and efficiently verifiable one-way puzzles $(\EVOWP)$.}\label{fig:graph}
\end{figure} 
\section{Technical Overview}

\subsection{A cleaner construction of EFI pairs from any one-way puzzle}\label{ssub:efiowp}

In a recent work by Khurana and Tomer~\cite{khurana2024commitments}, it was shown that there is a black-box construction of an EFI pair from any one-way puzzle. Since one-way puzzles can be built from one-way state generators, this then shows that if one-way state generators exist, so do EFI pairs (and thus quantum bit commitments).

Since EFI pairs intuitively are a ``pseudorandom" primitive while one-way puzzles are a ``one-way" primitive, the argument presented in~\cite{khurana2024commitments} is heavily inspired by the classical construction of a pseudorandom generator from any one-way function, first shown in~\cite{HILL99}. 

The key idea behind~\cite{HILL99} is to first use the one-way function to construct something called a pseudoentropy generator. A pseudoentropy generator is simply a samplable distribution which is indistinguishable from another (not necessarily samplable) distribution with greater entropy. Then, the pseudoentropy generator is used to construct a non-uniform PRG. That is, a PRG where the construction takes in a short advice string of length $O(\log \lambda)$ depending on the security parameter. This gives out a different PRG candidate for each possible value of the advice string. Applying a PRG combiner to all of these candidates then gives a standard PRG.

~\cite{khurana2024commitments} follows the same overall structure to build an EFI pair from any one-way puzzle. In particular, they show how to build a pseudoentropy generator from any one-way puzzle, and then show how to use a one-way puzzle to build something they refer to as an imbalanced EFID pair. An EFID pair is classical version of an EFI pair. We recall that a non-uniform EFID pair is an EFID pair that takes in a short advice string. An imbalanced EFID pair is a stronger primitive than a non-uniform EFID pair, where there are additional requirements on hiding and/or binding when the primitive is instantiated with incorrect advice. Finally, they show how to use an imbalanced EFID pair to build a standard EFI pair, although this technique requires switching to quantum output.

In this work, we present an alternate construction of EFI pairs from one-way puzzles, with several advantages. The foremost advantage, which is useful for our other results, is that our construction works even when instantiated with weak one-way puzzles. In addition to this, the proof of our construction is significantly simpler, and relies almost entirely on standard classical techniques.

\begin{theorem}[Informal version of~\Cref{cor:weakowptoefi}]
    If there exists a weak one-way puzzle, then there exists an EFI pair.
\end{theorem}

\paragraph*{The overall approach.} While~\cite{khurana2024commitments} relies on the techniques of~\cite{HILL99} to realize their construction, there have been a number of follow-up works succeeding~\cite{HILL99} providing more efficient constructions of PRGs from OWFs~\cite{HRV10} ~\cite{VZ12,MP23}. In particular, we observe that the techniques of~\cite{VZ12} are particularly ``quantum-friendly," much more so than the techniques of~\cite{HILL99}. Furthermore, we make the (as far as we are aware) novel observation that the construction of~\cite{VZ12} gives a pseudorandom generator even when instantiated with a weak one-way function.

\paragraph*{The failure of Goldreich-Levin for weak one-way puzzles.}
One key idea underlying~\cite{HILL99}, as well as most other constructions of PRGs from one-way functions~\cite{VZ12,MP23}, is to extract $H_{min}(x | f(x)) + O(\log n)$ bits of pseudoentropy from $x$ given $f(x)$. The leftover hash lemma gives the ability to extract $H_{min}(x | f(x)) - O(\log n)$ bits of entropy from $x$, and Goldreich-Levin provides an extra $O(\log n)$ bits of pseudoentropy~\cite{GL89}, so these two techniques together can extract a pseudorandom string of length $H_{min}(x | f(x)) + O(\log n)$ from $x$ given $f(x)$.

In particular, the Goldreich-Levin theorem shows that if there is an algorithm distinguishing $\Ext(x)$ from uniform given $f(x)$ with advantage $\epsilon$, then there is an algorithm computing $x$ from $f(x)$ with probability $\poly(\epsilon)$~\cite{GL89}. Since $\epsilon^2$ is negligible for a strong one-way function, so is $\epsilon$, and so these distributions are indistinguishable. However, if $f$ is only a weak one-way function, then we only get a constant bound on the distinguishing advantage, and so the approaches of~\cite{HILL99,VZ12,MP23} all break down.

A similar approach, with some technically involved adjustments to handle quantum sampling, is done in~\cite{khurana2024commitments} by using a quantum version of the Goldre-ich-Levin theorem~\cite{AC01}. In particular,~\cite{khurana2024commitments} also relies on using Goldreich-Levin to extract $O(\log n)$ from the key $k$ given the puzzle $s$. But for the same reason as before, this approach does not hold when the sampler is only weakly one-way.

Furthermore, there is a lot of technical care needed when using the leftover hash lemma and Goldreich-Levin on puzzles sampled using quantum randomness~\cite{khurana2024commitments}. This is because the pre-image space of a puzzle is now a distribution over keys instead of a set, and so hashing techniques become significantly more complicated. Luckily,~\cite{VZ12} demonstrates a way to construct PRGs from one-way functions without relying on either of these techniques, providing an approach that is both quantum-friendly and applies even with weak security. We adapt their techniques to give a construction of EFI pairs from weak one-way puzzles illustrated in~\Cref{fig:hrvconst}.

\begin{figure}
    \centering
    \begin{tikzpicture}[scale=0.75]
        \node (a) at (5,1.6) {$\vdots$};

        \draw (-0.5,0) -- (-0.5,1);
        \node (s1) at (0.5,0.5) {$s_{a1}$};
        \draw (1.5,0) -- (1.5,1);
        \node (k1) at (2.5,0.5) {$k_{a1}$};
        \draw (3.5,0) -- (3.5,1);
        \node (s2) at (4.5,0.5) {$s_{a2}$};
        \draw (5.5,0) -- (5.5,1);
        \node (k2) at (6.5,0.5) {$k_{a2}$};
        \draw (7.5,0) -- (7.5,1);
        \node (s2) at (8.5,0.5) {$s_{a3}$};
        \draw (9.5,0) -- (9.5,1);
        \node (k2) at (10.5,0.5) {$k_{a3}$};
        \draw (11.5,0) -- (11.5,1);

        \draw (-1.2,2) -- (-1.2,3);
        \node (s1) at (-0.2,2.5) {$s_{21}$};
        \draw (0.8,2) -- (0.8,3);
        \node (k1) at (1.8,2.5) {$k_{21}$};
        \draw (2.8,2) -- (2.8,3);
        \node (s2) at (3.8,2.5) {$s_{22}$};
        \draw (4.8,2) -- (4.8,3);
        \node (k2) at (5.8,2.5) {$k_{22}$};
        \draw (6.8,2) -- (6.8,3);
        \node (s2) at (7.8,2.5) {$s_{23}$};
        \draw (8.8,2) -- (8.8,3);
        \node (k2) at (9.8,2.5) {$k_{23}$};
        \draw (10.8,2) -- (10.8,3);

        \draw (-0.1,3) -- (-0.1,4);
        \node (s1) at (0.9,3.5) {$s_{11}$};
        \draw (1.9,3) -- (1.9,4);
        \node (k1) at (2.9,3.5) {$k_{11}$};
        \draw (3.9,3) -- (3.9,4);
        \node (s2) at (4.9,3.5) {$s_{12}$};
        \draw (5.9,3) -- (5.9,4);
        \node (k2) at (6.9,3.5) {$k_{12}$};
        \draw (7.9,3) -- (7.9,4);
        \node (s2) at (8.9,3.5) {$s_{13}$};
        \draw (9.9,3) -- (9.9,4);
        \node (k2) at (10.9,3.5) {$k_{13}$};
        \draw (11.9,3) -- (11.9,4);
    
        \draw (-0.5,0) -- (11.5,0);
        \draw (-0.5,1) -- (11.5,1);
        \draw (-1.2,2) -- (10.8,2);
        \draw (-1.2,3) -- (11.9,3);
        \draw (-0.1,4) -- (11.9,4);
        
        \draw[dashed] (1,-0.5) -- (1,4.5);
        \draw[dashed] (9,-0.5) -- (9,4.5);

        \draw (1.3, 2) circle[x radius = 0.2cm, y radius=3cm];
        \draw (1.3, -1.1) -- (1.3, -2);
        \node (ext) at (1.3, -2.2) {$\Ext$};
        \draw[->] (1.3, -2.5) -- (1.3, -3.4);
        \node (r) at (1.3, -3.6) {$r_1$};
        
        \draw (1.9, 2) circle[x radius = 0.2cm, y radius=3cm];
        \draw (1.9, -1.1) -- (1.9, -2);
        \node (ext) at (1.9, -2.2) {$\Ext$};
        \draw[->] (1.9, -2.5) -- (1.9, -3.4);
        \node (r) at (1.9, -3.6) {$r_2$};
        
        \draw (8.7, 2) circle[x radius = 0.2cm, y radius=3cm];
        \draw (8.7, -1.1) -- (8.7, -2);
        \node (ext) at (8.7, -2.2) {$\Ext$};
        \draw[->] (8.7, -2.5) -- (8.7, -3.4);
        \node (r) at (8.7, -3.6) {$r_d$};

        \node (dots) at (5,-2.2) {$\dots$};
        \node (dots) at (5,-3.6) {$\dots$};
        
    \end{tikzpicture}
    \caption{The construction of~\cite{VZ12,HRV10} applied to a one-way puzzle $\Samp \to (k,s)$. The idea is that many samples are taken and arranged in a grid. Then, each row is given a random offset, with both sides truncated. Finally, some number of random bits are extracted from each column using a pairwise-independent hash $\Ext$. This produces a pseudorandom string with less than full entropy, and we can repeat to get a non-uniform EFID pair.}
    \label{fig:hrvconst}
\end{figure}
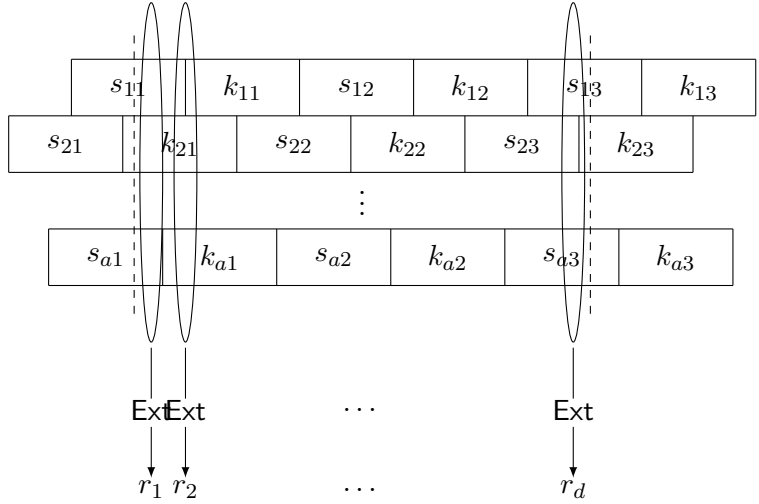

\paragraph*{The construction of~\cite{VZ12}.} To build a PRG from a one-way function $f$,~\cite{VZ12} makes the observation that the distribution $(f(x),x)$ satisfies a property which they call KL-hard to sample. In particular, this means that for any sampler $\S$ (which in this case can be thought of as a distributional inverter),
$$KL(f(x),x || f(x), \S(f(x))) \geq \delta$$
for some value of $\delta \geq \frac{1}{\poly(\lambda)}$. Here $KL$ refers to Kullback-Leibler divergence, or ``relative entropy." They then adapt the techniques of~\cite{HRV10} to build a PRG from a distribution which is KL-hard to sample. Note that this construction requires knowledge of the entropy of the KL-hard to sample distribution. However, for a one-way function, $H(f(x),x) = |x|$ the input length of the one-way function.

For an $\epsilon$ one-way function, the KL-hardness parameter is $\delta = -\log \epsilon$. Thus, for a standard one-way function, $\delta = \omega(\log \lambda)$. But the techniques of~\cite{HRV10} apply whenever $\delta = \frac{1}{\poly(\lambda)}$, and so the techniques of~\cite{VZ12} work just as well for weak one-way functions. Thus, the same construction gives a PRG from any weak one-way functions.

\paragraph{Building a KL-hard to sample distribution from a one-way puzzle}
The key observation underlying~\cite{VZ12} is that KL divergence can only decrease from computation. That is, for any function $F$,
$$KL(F(X) || F(Y)) \leq KL(X || Y)$$
But the boolean function $F(y,x) = 1$ if and only if $f(x) = y$ is well-defined. So if $\S$ is any sampler,
$$KL(f(x),x || f(x),\S(f(x))) \geq KL(F(f(x),x) || F(f(x),\S(f(x))))$$
$$= KL(1 || Bern(p)) = -\log p$$
where $p$ is the advantage of $\S$ in the one-way function security game. This immediately gives that the distribution $(f(x),x)$ is KL-hard to sample.

We observe that the same exact technique also works for one-way puzzles. In particular, let $(\Samp,\Ver)$ be a one-way puzzle and let $\Samp \to (k,s)$. The equivalent of checking if $f(x) = y$ is simply to run verification. And so
$$KL(s,k || s, \S(k)) \geq KL(\Ver(s,k) || \Ver(s, \S(k))) = KL(Bern(q) || Bern(p))$$
where $p$ is the success probability of $\S$ in the one-way puzzle game and $q$ is the correctness parameter of the one-way puzzle. Although we do not have $KL(Bern(q) || Bern(p)) = -\log p$, when $\Samp$ is a weak one-way puzzle, we can still lower bound $KL(s,k || s, \S(k))$ by $\frac{1}{\poly(\lambda)}$. And so $(s,k)$ is KL-hard to sample.

\paragraph{Building a non-uniform EFID pair from a KL-hard to sample distribution}
Note that the techniques~\cite{VZ12} uses to build a PRG from a KL-hard to sample distribution are entirely black box, and so also work in the quantum setting. Thus, applying the same construction to $(s,k)$ produces a pseudorandom distribution $D$ with length $d = \abs{D}$ depending on $H(k,s)$. When building a PRG, the approach~\cite{HRV10,VZ12} take is to argue that $D$ can be sampled by applying some function $G$ to a uniformly random string of length $d' < d$, and so $G$ is a PRG. Here, the randomness of the distribution is quantum, and so this idea will not apply directly. But similar reasoning can be used to show a upper bound on the entropy of $D$. In particular, we produce such an argument directly and show that $H(D) < d - \poly(\lambda)$. For a visualization of the construction of $D$, see~\Cref{fig:hrvconst}.

We then observe that any distribution with sufficiently less entropy than its length must have some statistical distance from the uniform distribution. Thus, the~\cite{VZ12} construction applied to a one-way puzzle produces a distribution $D$ which is indistinguishable from uniform but has noticeable statistical distance from uniform. We then use parallel repetition to boost the statistical distance to $1 - \negl(\lambda)$, and so the pair $(\U^{t},D^t)$ forms a EFID pair.

Unfortunately, this construction has a number of pseudorandom bits dependent on $H(k,s)$. Thus, the EFID pair construction has to have knowledge of the entropy of the one-way puzzle sampler output. This can be done by giving the construction $\Theta(\log\lambda)$ bits of advice, and so instead of a full EFID pair, we get a non-uniform EFID pair.

\paragraph{From non-uniform EFID pairs to EFI pairs}
To recap,~\cite{khurana2024commitments} built imbalanced EFID (a stronger version of non-uniform EFID) from one-way puzzles, while our technique only builds non-uniform EFID from one-way puzzles. Note that this is not a fundamental difference, upon observation it is clear that our construction also satisfies the requirements of imbalanced EFID. 

However, the reason~\cite{khurana2024commitments} required this stronger notion of non-uniform EFID was because, at the time that work was published, it was unknown how to build combiners for EFI pairs. Recent work (interestingly using similar techniques to~\cite{khurana2024commitments}) has shown how to combine EFI pairs~\cite{HKNY23}, and so using these techniques EFI pairs follow directly from non-uniform EFID.

\subsection{Combiners and universal constructions}
One major property satisfied by one-way functions is the existence of a universal construction~\cite{Levin87}. By this, we mean that there exists a specific construction of a one-way function which is secure if any one-way functions exist.

As shown originally by Levin~\cite{Levin87} and formalized in~\cite{HKNRR05}, this useful fact is essentially a corollary of the fact that there exists \emph{robust combiners} for one-way functions. That is, given any two one-way function candidates $f$ and $g$, there is a construction $h^{f,g}$ such that $h$ is one-way as long as one of $f$ or $g$ is one-way.

The universal one-way function is then defined as follows. Take the first $\log \lambda$ Turing machines and treat them as one-way function candidates. Running the combiner on all these candidates results in a universal one-way function $f_U$. As long as one-way functions exist, there is a Turing machine with some constant length which acts as a good one-way function. Thus, for all sufficiently large $\lambda$, $f_U$ will also be a one-way function.

Since both combiners and universal constructions are highly desirable properties, we would like to investigate whether robust combiners also exist for one-way puzzles. We thus prove the following theorem

\begin{theorem}[Informal version of~\Cref{cor:owpcomb}]
    There exists a robust combiner for one-way puzzles.
\end{theorem}

with the following corollary

\begin{corollary}[Informal version of~\Cref{thm:uowp}]
    There exists a pair of algorithms $(\Samp_U,\Ver_U)$ such that as long as one-way puzzles exist, $(\Samp_U,\Ver_U)$ is a one-way puzzle.
\end{corollary}

Note that it has been shown that combiners and universal constructions exist for quantum primitives which both imply one-way puzzles and are implied by one-way puzzles, namely one-way state generators and EFI pairs respectively~\cite{HKNY23}. Thus, this result should not be particularly surprising. However, none of the arguments for constructing combiners for one-way state generators, EFI pairs, or one-way functions translate directly into building combiners for one-way puzzles.

Note that if we know that both candidate one-way puzzles satisfy correctness, then it is easy to construct a combiner. In particular, running both candidate samplers in parallel and having the verification algorithm accept if and only if both candidate verification algorithms accept is enough to ensure that the combined construction satisfy both correctness and security.

However, if we omit the correctness requirement, then it is possible that the ``bad" verification algorithm always rejects. In this case, the combiner we defined previously will also not satisfy correctness.

To resolve this issue, we follow the template of~\cite{HKNY23} and show that there is a ``correctness guaranteeing" procedure for any one-way puzzle. Namely
\begin{theorem}[Informal version of~\Cref{cor:owpcorrect}]
    Let $(\Samp,\Ver)$ be a one-way puzzle candidate. There exists a construction $(\Samp',$ $\Ver')$ where $\Samp',\Ver'$ depend on $(\Samp,\Ver)$ satisfying the following
    \begin{enumerate}
        \item If $(\Samp,\Ver)$ is a one-way puzzle, then so is $(\Samp',\Ver')$.
        \item Regardless of whether $(\Samp,\Ver)$ is a one-way puzzle, $(\Samp',\Ver')$ satisfies one-way puzzle correctness.
    \end{enumerate}
\end{theorem}

If we apply this correctness amplification procedure to the candidate one-way puzzles and then apply the security combiner described earlier, we achieve a robust combiner for one-way puzzles.

The main question remaining is how to actually do this correctness amplification. The natural approach to correctness guaranteeing (which is analogous to the approach used by~\cite{HKNY23}) is to have the sampler check whether verification passes on its produced key-puzzle pair. If not, the sampler will output a special symbol $\bot$, on which the verifier will always accept. However, this approach requires that the sampler be able to run the verifier. But for one-way puzzles, the verification algorithm may not be efficient. 

Our solution is to defer the checking step to the verification algorithm itself. In particular, we will say that a puzzle is good if the probability that verification passes when it is naturally generated is high. Since our verification algorithm is inefficient, it has the computational resources to check if a given puzzle is good. The key idea, then is that we modify verification to automatically accept any good puzzle.

If the scheme originally satisfied correctness, then all but a negligible fraction of puzzles will be good and so we do not compromise security. Furthermore, since the probability that verification fails on a good puzzle is by definition low, the probability that the modified verification fails will also be low. Note that this style of correctness guaranteeing will only give a guarantee that the correctness error is below some constant (say $1/2$). We can then boost to full correctness through parallel repetition.

\paragraph*{A note on the definition of one-way puzzles} The original definition of one-way puzzles introduced required that the verification procedure be represented by a Turing machine which is guaranteed to halt (i.e. a decider)~\cite{khurana2024commitments}. When defining a robust combiner for use in a universal construction, it is necessary that the combiner work even if one of the candidate verification algorithms does not halt. This makes building a combiner seemingly as difficult as solving the halting problem.

We instead define one-way puzzles so that verification can be any arbitrary function. Note that, because $\Ver$ is never actually run, all known constructions using one-way puzzles go through when using this weakened definition. In addition, under our definition, combiners and universal constructions exist. Thus, we believe that this generalized definition is the ``right" definition of a one-way puzzle, and that the restriction of verification to halting Turing machines used by~\cite{khurana2024commitments} is unnecessarily restrictive.

\subsection{Amplification of one-way puzzles}\label{ssub:amp}

A second desirable property for a central primitive to have is an amplification theorem. In particular, given a one-way function with a weaker security guarantee, it is possible to build a normal one-way function. This makes it significantly easier to construct one-way functions from other primitives as well as produce candidate one-way functions.

Thus, one may wonder whether the same is true for one-way puzzles. That is, given a one-way puzzle with a weakened security guarantee, is it possible to build a normal one-way puzzle? We can also ask the same question of correctness. Given a one-way puzzle with a weakened correctness guarantee, is it possible to build a normal one-way puzzle?

In particular, we define $(\alpha,\beta)$ one-way puzzles, where $\alpha$ is the correctness error and $\beta$ the security error. Observe that standard one-way puzzles are simply $(\negl(\lambda),\negl(\lambda))$ one-way puzzles. We show the following

\begin{theorem}[Restatement of~\Cref{thm:owpcorramp}]
    If there exists a $(1 - 1/\poly(\lambda),$ $\negl(\lambda))$ one-way puzzle, then there exists a $(\negl(\lambda),$ $ \negl(\lambda))$ one-way puzzle.
\end{theorem}
\begin{theorem}[Restatement of~\Cref{thm:owpamplification}]
    If there exists a $(\negl(\lambda),$ $1 - 1/\poly(\lambda))$ one-way puzzle, then there exists a $(\negl(\lambda),$ $\negl(\lambda))$ one-way puzzle.
\end{theorem}

\paragraph*{Amplifying Security}

For the purposes of this section, we will refer to a $(\negl(\lambda),$ $1 - 1/\poly(\lambda))$ one-way puzzle as a weak one-way puzzle. We will also refer to the standard notion of a one-way puzzle as a strong one-way puzzle. The question of security amplification can then be rephrased as ``can we build a strong one-way puzzle from any weak one-way puzzle?"

Recently,~\cite{bostanci2023efficient} showed that parallel repetition amplifies soundness guarantees for any $3$ round quantum interactive protocol. At first glace, one might think that this result immediately gives a security amplification theorem for one-way puzzles. 

Upon observation, it turns out that the argument of~\cite{bostanci2023efficient} relies on the assumption that the security game itself can be run efficiently. But the one-way puzzle security game requires running the verification algorithm, which has no guarantees on efficiency. And so the obvious approach to amplifying security falls short in this setting.

But what can we do? Our key observation is that strong one-way puzzles can be built from EFID pairs (which we recall is the classical version of an EFI pair). In addition,~\cite{khurana2024commitments} shows that one-way puzzles can be used to build EFI pairs, and along the way they show that strong one-way puzzles can be used to build a variant of EFID pairs, which we here call a non-uniform EFID pair. Unfortunately, their techniques do not work for weak one-way puzzles, an issue we remedy in~\Cref{ssub:efiowp}.

The outline of our argument is to use our improved construction of EFI pairs from one-way puzzles from~\Cref{ssub:efiowp}, which shows that we can build non-uniform EFID pairs from weak one-way puzzles as well. We then show how to build strong one-way puzzles from non-uniform EFID pairs.

\paragraph{Building strong one-way puzzles from non-uniform EFID}
Recall, a non-uniform cryptographic primitive is a cryptographic primitive where the construction takes in a short advice string of length $O(\log \lambda)$. Our construction of strong one-way puzzles from EFID pairs then allows us to build a non-uniform strong one-way puzzle from a non-uniform EFID pair. For each possible advice string $s$, we can instantiate the non-uniform strong one-way puzzle with $s$ to get a new strong one-way puzzle candidate. For each security parameter, one of these candidates is a strong one-way puzzle. Thus, using a combiner on all of these candidates simultaneously produces a strong one-way puzzle which does not need any advice string. 

We remark that in general, if we have a robust combiner for a primitive then we can turn any non-uniform construction of that primitive into its full version.

\paragraph*{Amplifying Correctness}

To amplify correctness, we observe that our correctness guaranteer will always increase correctness at some cost to security. By carefully tracking this cost and interleaving with security amplification, it is possible to boost to full correctness without hurting security.

\subsection{Relationships with other QCCC primitives}

It has been shown in~\cite{khurana2024commitments} that one-way puzzles can be built from almost all QCCC style primitives. In particular, they show how to build a one-way puzzle from a digital signature, a symmetric encryption protocol, or a commitment scheme. Just as a one-way function can be built from any useful classical primitive, a one-way puzzle can be built from any useful QCCC primitive.

Note that minimality of a primitive is not very hard to achieve. As an example, any primitive which can be built unconditionally is at least as minimal as a one-way function. But, importantly, one-way functions are also \emph{useful}. One-way functions can be used to construct a large class of important cryptographic primitives, often referred to as symmetric key or Minicrypt primitives. In particular, one-way functions can be used to build (classical) digital signatures, symmetric encryption protocols, and commitment schemes. Thus, if we want to treat one-way puzzles as a central primitive for QCCC cryptography, it seems important that one-way puzzles imply at least one more directly useful QCCC primitive.

Unfortunately, existing results in this direction are noticeably weaker. In particular, as far as the authors are aware, the only QCCC primitive for which a construction is known from one-way puzzles is non-uniform EFID pairs~\cite{khurana2024commitments}. Instead, there are \emph{quantum output} implications of one-way puzzles. In particular, it is known how to build commitments with quantum output (and all equivalent primitives) from one-way puzzles~\cite{khurana2024commitments}.

\subsection{Efficiently verifiable one-way puzzles}

The key challenge to using one-way puzzles to build other primitives is that the constructions may not make use of the verification scheme in a black-box manner since the verification scheme is not itself efficient. Thus, to make the problem easier, we consider a variant of one-way puzzles with efficient verification. 

A very similar primitive, termed ``hard quantum planted problems for QCMA," has been studied before in the context of publicly verifiable deletion~\cite{KNY23}. A hard quantum planted problem is essentially an efficiently verifiable one-way puzzle with perfect correctness. Direct observation shows that perfect correctness is unnecessary for any of the applications of hard quantum planted problems for QCMA, and so their results hold for efficiently verifiable one-way puzzles as well. In particular, they show that
\begin{theorem}[Theorem 6.2 from~\cite{KNY23}]
    If there exists an efficiently verifiable one-way puzzle and quantum
    $$Z \in \{SKE, COM, PKE, ABE, QFHE, TRE, WE\},$$
    then there exists $Z$ with publicly verifiable deletion.
\end{theorem}

Their construction requires building a stronger variant of a one-time signature scheme from efficiently verifiable one-way puzzles. In particular
\begin{theorem}[Theorem 3.2 from~\cite{KNY23}]
    If there exists an efficiently verifiable one-way puzzle, then there exists a QCCC one-time signature scheme.
\end{theorem}

The construction is essentially just a Lamport signature~\cite{lamport1979constructing}. As this theorem is not presented with full proof details in~\cite{KNY23}, for completeness we restate this claim as~\Cref{thm:otsowp} and give a full proof.

Since efficiently verifiable one-way puzzles seem more useful than normal one-way puzzles, we might wonder whether they are also minimal. Fortunately, most of the constructions of one-way puzzles from QCCC primitives have efficient verification algorithms. The two notable exceptions are EFID pairs and commitment schemes. 
\begin{theorem}[Theorems A.4 and A.6 from~\cite{khurana2024commitments} and~\Cref{thm:nicom,thm:qprgtoowp} in this paper]
    If there exists a QCCC signature scheme, secret key encryption scheme, non-interactive commitment scheme, or pseudodeterministic PRG, then there exists an efficiently verifiable one-way puzzle.
\end{theorem}

Applying this theorem to the results of~\cite{KNY23} then gives the following two interesting corollaries.

\begin{corollary}
    If there exists QCCC
    $$Z \in \{SKE,PKE,ABE,QFHE,TRE\}$$
    then there exists $Z$ with publicly verifiable deletion.
\end{corollary}

\begin{corollary}
    There exists an efficiently verifiable one-way puzzle if and only if there exists a QCCC one-time signature scheme.
\end{corollary}

\paragraph*{Amplification and combiners for efficiently verifiable one-way puzzles.} Since efficiently verifiable one-way puzzles seem to be about as minimal for QCCC as one-way puzzles, but have much more powerful applications, we may consider whether efficiently verifiable one-way puzzles should instead be considered a ``central" primitive for QCCC cryptography. We then may hope that there exists an amplification theorem and a universal construction for efficiently verifiable one-way puzzles. We show that this is indeed the case.

\begin{theorem}[Restatement of~\Cref{thm:owpcorramp}]
    If there exists a $(1 - 1/\poly(\lambda),$ $\negl(\lambda))$ efficiently verifiable one-way puzzle, then there exists a $(\negl(\lambda),\negl(\lambda))$ efficiently verifiable one-way puzzle.
\end{theorem}
\begin{theorem}[Restatement of~\Cref{thm:secampevowp}]
    If there exists a $(\negl(\lambda),$ $1 - 1/\poly(\lambda))$ efficiently verifiable  one-way puzzle, then there exists a $(\negl(\lambda),$ $\negl(\lambda))$ efficiently verifiable one-way puzzle.
\end{theorem}
\begin{theorem}[Informal version of~\Cref{cor:owpcomb}]
    There exists a robust combiner for efficiently verifiable one-way puzzles.
\end{theorem}
\begin{corollary}[Informal version of~\Cref{thm:universalevowp}]
    There exists a pair of algorithms $(\Samp_U,$ $\Ver_U)$ such that as long as efficiently verifiable one-way puzzles exist, $(\Samp_U,\Ver_U)$ is an efficiently verifiable one-way puzzle.
\end{corollary}

Note that most of the barriers to these results go away when the verification algorithm is required to be efficient. Thus, the ``naive" constructions described earlier are provably secure for efficiently verifiable one-way puzzles.

\paragraph*{Are one-way puzzles equivalent to efficiently verifiable one-way puzzles?} Although the advantage of treating efficiently verifiable one-way puzzles as a ``central" QCCC primitive is that it has actual applications in the QCCC setting, this does come at a cost to its ``minimality''. It is not clear how to build efficiently verifiable one-way puzzles from every primitive known to imply $\OWP$. In particular, constructions are lacking from EFID pairs and commitments.

Thus, we may consider whether or not it even matters whether verification is efficient. Ideally, we would be able to build an efficiently verifiable one-way puzzle from any one-way puzzle. In fact, if we restrict the sampling algorithm to being a classical randomized algorithm, such a claim holds true. Given a classical one-way puzzle, we can build an efficiently verifiable one-way puzzle by replacing the key with the random coins of the sampler. Then, the verifier can simply check whether running the sampler on the randomness given produces the given puzzle.

However, as this approach directly uses the randomness of sampling, it is inherently non-quantum. In fact, it turns out that in the quantum setting, there is a black-box separation

\begin{theorem}[Informal version of~\Cref{thm:separation}]\label{thm:sepintro}
    There exists a quantum oracle $\mathcal{O}$ relative to which one-way puzzles exist but efficiently verifiable one-way puzzles do not exist.
\end{theorem}

This theorem follows from a simple observation. A search-to-decision argument shows that any efficiently verifiable one-way puzzle can be broken using a QCMA oracle. But there exists an oracle relative to which pseudorandom states exist and BQP=QCMA~\cite{Kretschmer21Quantum}. As pseudorandom states can be used to build one-way puzzles~\cite{khurana2024commitments,morimae2022quantum},~\Cref{thm:sepintro} follows.

\paragraph*{A barrier against length shrinking for pseudorandom states} An open question in the literature is whether pseudorandom states with output length $n(\lambda)$ can be built from pseudorandom states with output length $n'(\lambda)$ for any values of $n,n'$ such that $n\neq n' \geq \log n$. However, pseudorandom states with output length $O(\log \lambda)$ can be used to build QCCC pseudodeterministic PRGs, and thus efficiently verifiable one-way puzzles. But our argument gives a black-box separation between efficiently verifiable one-way puzzles and pseudorandom states with output length $\lambda$. Thus, we get the following corollary
\begin{corollary}[Informal version of~\Cref{cor:prsseparation}]
    There exists a quantum oracle $\mathcal{O}$ relative to which PRSs with output length $\lambda$ exist but PRSs with output length $c\log \lambda$ (for $c > 12$) do not.
\end{corollary}

Note that this observation at its core comes from the simple observation that pseudodeterministic PRGs can be broken with a QCMA oracle, and so this observation is little more than a corollary of the results of~\cite{ALY23,Kretschmer21Quantum}, and is known in folklore. However, we provide a full formal proof of this statement as a contribution towards the systemization of knowledge in quantum cryptography.

\subsection{Equivalence to variant definitions}

\paragraph*{Random input one way puzzles} Another natural variant of one-way puzzles we might consider is a one-way puzzle where the key must be sampled uniformly at random, and then the puzzle is sampled from the key. We will call this a random input one-way puzzle. This more closely aligns with the classical notion of one-way functions, and in fact the construction of one-way puzzles from one-way state generators produces a random input one-way puzzle (assuming the key generation for the one-way state generator is uniform)~\cite{khurana2024commitments}.

~\cite{khurana2024commitments} left as an open question whether random input one-way puzzles can be built from arbitrary one-way puzzles. Note that this statement does hold classically, since both are equivalent to one-way functions.

We show that these two notions are indeed equivalent
\begin{theorem}[Restatement of~\Cref{thm:randominput}]
    If there exists a one-way puzzle, then there exists a random input one-way puzzle. If there exists an efficiently verifiable one-way puzzle, then there exists a random input efficiently verifiable one-way puzzle. 
\end{theorem}

The idea is fairly natural. We simply treat the random input as a one-time pad, and apply it to the original key. We then include the one-time padded key with the original puzzle in the final output.

Note that our amplification lemma for one-way puzzles also produces a random input one-way puzzle, and so also gives an indirect proof of this theorem, although this approach does not hold for efficiently verifiable one-way puzzles.

\paragraph*{Distributional one way puzzles} We can also consider an analogue of distributional one way functions, which are known to be equivalent to one way functions. That is, we say that a puzzle is distributionally one way if, given a puzzle, it is hard to sample the conditional distribution over keys. Note that every one way puzzle is a distributional one way puzzle. Interestingly, distributional one way puzzles do not need a verification algorithm, since their security comes from the sampled distribution.

We show that distributional one way puzzles are equivalent to one way puzzles. This follows directly from our techniques for amplification. In particular, a distributional one way puzzle is KL-hard to sample, and so can thus be used to build strong one way puzzles.

\begin{theorem}
    If there exists a distributional one-way puzzle, then there exists a one way puzzle.
\end{theorem}

We remark that as there is no verifier, there is not a natural definition of an "efficiently verifiable" distributional one-way puzzle. 

\section{Open Questions}

Although we are aware of a few implications, the landscape of QCCC reductions, even those relating to one-way puzzles/efficiently verifiable one-way puzzles, is still fairly unexplored. We list a few interesting questions in this space related to our work
\begin{enumerate}
    \item Is it possible to build efficiently verifiable one-way puzzles from a QCCC commitment scheme? QCCC commitments and EFID pairs are the two QCCC primitives for which the obvious construction of one-way puzzles does not have an efficient verifier. If the answer to this question is no, then it may be possible to build QCCC commitments from standard one-way puzzles.
    \item Are there any useful cryptographic primitives we can construct from one-way puzzles without efficient verification? Due to the black-box separation between one-way puzzles and efficiently verifiable one-way puzzles, it seems like the answer may be no. However, a few primitives (such as EFID pairs and QCCC commitments) fall outside of this separation, and so there is still hope for a construction.
    \item Is there a combiner for QCCC EFID pairs? If so, by using the construction of non-uniform EFID from one-way puzzles, we would be able to construct standard EFID from one-way puzzles. Interestingly, there does exist a combiner for both the quantum version and the classical version of this primitive~\cite{HKNY23,Levin87,Goldreich90}.
    \item Can we build any QCCC primitives besides one-time signatures from efficiently verifiable one-way puzzles, for example secret key encryption or pseudodeterministic PRGs? What about many-time signatures? Although we observe that the known construction of many time signatures from one-way functions uses a pseudorandom function in order to be stateless, so this may necessitate building a QCCC style pseudorandom function~\cite{GMR87}.
\end{enumerate} 
\section{Preliminaries}

\subsection{Definitions of QCCC primitives}
As discussed previously this definition of $\OWP$ slightly generalizes the notion given in \cite{khurana2024commitments}.

\begin{definition}
An $(\alpha,\beta)$ one way puzzle ($\OWP$) is a pair of a sampling algorithm and a verification function $(\Samp,\Ver)$ with the following syntax:
    \begin{enumerate}
        \item $\Samp(1^{\lambda}) \to (k,s)$ is a uniform QPT algorithm which outputs a pair of classical strings $(k,s)$. We refer to $s$ as the puzzle and $k$ as the key. Without loss of generality, we can assume $k \in \{0,1\}^\lambda$.
        \item $\Ver(k,s) \to b$ is some (possibly uncomputable) function which takes in a key and puzzle and outputs a bit $b \in \{0,1\}$.
    \end{enumerate}
    satisfying the following properties:\\
    \begin{enumerate}
        \item Correctness: For all sufficiently large $\lambda$, outputs of the sampler pass verification with overwhelming probability
        $$\Pr_{\Samp(1^{\lambda})\to (k,s)}[\Ver(k,s)\to 1] \geq 1 - \alpha$$
        \item Security: Given a puzzle $s$, it is computationally infeasible to find a key $s$ which verifies. That is, for all non-uniform QPT algorithms $\A$, for all sufficiently large $\lambda$,
        $$\Pr_{\Samp(1^{\lambda}) \to (k,s)}[\Ver(\A(s),s) \to 1] \leq \beta$$
    \end{enumerate}
    If for all $c$, $(\Samp,\Ver)$ is a $(\lambda^{-c},\lambda^{-c})$ one way puzzle, then we say that $(\Samp,\Ver)$ is a strong $\OWP$ and omit the constants. When unambigious, we will simply say that $(\Samp,\Ver)$ is a $\OWP$.
    
\end{definition}

    

\begin{definition}
    A one-time signature scheme is a set of QPT algorithms $(KeyGen,$ $S, V)$ with the following syntax
    \begin{enumerate}
        \item $KeyGen(1^\lambda)\to (vk,sk)$ takes the security parameter as input and ouptuts a signing key $sk$ and a verification key $vk$
        \item $S(sk,m) \to \sigma$ takes in the signing key and a message as input, and outputs a signature $\sigma$
        \item $V(vk,m,\sigma) \to 0/1$ takes in a verification key $vk$, a message $m$, and a signature $\sigma$, and outputs a single bit
    \end{enumerate}
    satisfying the following security properties
    \begin{enumerate}
        \item Correctness: For all $m$ in the message space,
        $$\Pr_{KeyGen(1^\lambda) \to (vk,sk)}[V(vk,m,S(sk,m)) \to 1] \geq 1 - \negl(\lambda)$$
        \item One-time Security: An adversary with the ability to make one signature query can not forge a signature for a different message. More formally, for all $m_0 \neq m_1$ in the message space and for all PPT $\A$,
        $$\Pr_{KeyGen(1^\lambda)\to (vk,sk)}[V(vk,m_1,\A(vk,S(sk,m_0))) \to 1] \leq \negl(\lambda)$$
    \end{enumerate}
\end{definition}

\begin{definition}[Pseudodeterministic Quantum Pseudorandom Generator \cite{ananth2023pseudorandom}]
    A pseudodeterministic quantum pseudorandom generator $(\QPRG)$ is a uniform QPT algorithm $G$ that on input a classical seed $k \in \{0,1\}^{n(\lambda)}$ outputs a string of length $\ell(\lambda)$ with the following guarantees:
    \begin{enumerate}
        \item Pseudodeterminism: there exists a constant $c > 0$ and a function $\mu(\lambda) = O(\lambda^{-c})$ such that for every $\lambda \in \N$, there exists a set of good seeds $\mathcal{K}_\lambda \subseteq \{0,1\}^\lambda$ satisfying
        \begin{equation*}
            \Pr_{\{0,1\}^\lambda \to k}[k \in \mathcal{K}_\lambda] \geq 1 - \mu(\lambda)
        \end{equation*}
        \begin{equation*}
            \forall\ k\in \mathcal{K}_\lambda,\ \max_{y \in \{0,1\}^{\ell(\lambda)}}\Pr[y = G_\lambda(k)] \geq 1 - \mu(\lambda)
        \end{equation*}
        \item Stretch: $\ell(\lambda) > n(\lambda)$
        \item Security: For every non-uniform QPT algorithm $\A$,
        $$\abs{\Pr_{\{0,1\}^\lambda \to k}[\A(G(k)) \to 1] - \Pr_{\{0,1\}^{\ell(\lambda)} \to y}[\A(y) \to 1]} \leq \negl(\lambda)$$
    \end{enumerate}
\end{definition}

\begin{definition}[Commitment scheme from~\cite{khurana2024commitments}]
    A commitment scheme is an efficient two-party protocol between a committer $Com$ and a receiver $Rec$ consisting of a commit stage and an opening stage operating on a private input $m$ described as follows
    \begin{enumerate}
        \item Commit stage: both parties receive a unary security parameter $1^\lambda$. The committer $Com$ receives a private input $m$. It interacts with the receiver $Rec$ using only classical messages, and together they produce a transcript $z$. At the end of the stage, both parties hold a private quantum state $\rho_{Com}$ and $\rho_{Rec}$ respectively.
        \item (Non-interactive) opening stage: both parties receive the transcript $z$ as well as their private quantum states $\rho_{Com}$ and $\rho_{Rec}$ respectively. The committer $Com$ sends a single message $d$ to the receiver $Rec$. At the end of the stage, the receiver either outputs a message or the reject symbol $\bot$.
    \end{enumerate}
    satisfying the following two properties
    \begin{enumerate}
        \item Correctness: For all messages $m$, when $Com$ and $Rec$ interact honestly, the probability that $Rec$ outputs $m$ at the end of the opening stage is at least $1 - \negl(\lambda)$.
        \item (Computational) hiding: For all $m\neq m'$ and for all QPT adversarial receivers $Rec'$, the transcript of the interaction between the adversarial receiver and the committer with input $m$ is indistinguishable from the transcript of the interaction between the adversarial receiver and the committer with input $m'$. That is,
        $$Com(m)\rightleftarrows Rec' \approx Com(m') \rightleftarrows Rec'$$
        \item (Computational weak honest) binding: For all $m$ and for all QPT adversarial senders $Com'$, the probability that $Com'$ wins the following game is $\leq \negl(\lambda)$
        \begin{enumerate}
            \item In the first stage, an honest receiver $Rec$ interacts with the honest committer $Com$ to produce a transcript $z$ and receiver state $\rho_{Rec}$
            \item In the second stage, the honest receiver $Rec$ is given $\rho_{Rec}$ and $z$, while $Com'$ is given $z$ (but not $\rho_{Com}$). They then proceed to run the opening stage with the committer replaced by $Com'$, and $Rec$ produces a final output $m'$. $Com'$ wins if $m' \neq m$ and $m' \neq \bot$. 
        \end{enumerate}
    \end{enumerate}

    If the receiver never sends any messages in either stage, then we say the commitment scheme is non-interactive. In this case, we write $Com(m)\to (c,d)$ where $c$ is the message sent in the first round (the commitment) and $d$ is the message sent in the second round (the decommitment). We then describe the final output of $Rec$ by $Rec(c,d) \to m'$.
\end{definition}

\begin{definition}
    An EFID pair is a randomized algorithm $\Gen(1^\lambda,b)$ taking a unary security parameter $\lambda$ and a classical bit $b\in\{0,1\}$ which outputs a classical string satisfying the following two properties:
    \begin{enumerate}
        \item Statistically far: $$\Delta(\Gen(1^\lambda, 0), \Gen(1^\lambda, 1)) \geq 1 - \epsilon$$
        \item Computationally close: For all QPT $\A$ and for all sufficiently large $\lambda$, the distributions $\Gen(1^\lambda,0)$ and $\Gen(1^\lambda, 1)$ are indistinguishable.
    \end{enumerate}

    If $\Gen$ is a quantum algorithm (with classical output), then we call $\Gen$ a quantum EFID pair (or QEFID).
\end{definition}

\subsection{Complexity}

\begin{definition}
    We say a promise problem $\Pi:\{0,1\}^*\to \{0,1,\bot\}$ is in $\mathsf{Promise}$ $\mathsf{QCMA}$ if there exists a QPT algorithm $\V(x,y)$ and a polynomial $p$ such that:
    \begin{enumerate}
        \item Completeness: If $\Pi(x) = 1$, then there exists a $p(|x|)$-bit string $y$ such that 
        $$\Pr[\V(x,y) \to 1] \geq \frac{2}{3}$$
        \item Soundness: If $\Pi(x) = 0$, then for all $p(|x|)$-bit strings $y$,
        $$\Pr[\V(x,y)\to 1] \leq \frac{1}{3}$$
    \end{enumerate}
\end{definition}

\begin{definition}
    We say a promise problem $\Pi:\{0,1\}^*\to \{0,1,\bot\}$ is in $\mathsf{Promise}$ $\mathsf{QMA}$ if there exists a QPT algorithm $\V(x,\ket{\phi})$ and a polynomial $p$ such that:
    \begin{enumerate}
        \item Completeness: If $\Pi(x) = 1$, then there exists a $p(|x|)$-qubit state $\ket{\phi}$ such that 
        $$\Pr[\V(x,\ket{\phi}) \to 1] \geq \frac{2}{3}$$
        \item Soundness: If $\Pi(x) = 0$, then for all $p(|x|)$-qubit states $\ket{\phi}$,
        $$\Pr[\V(x,\ket{\phi})\to 1] \leq \frac{1}{3}$$
    \end{enumerate}
\end{definition}

\subsection{Oracles}

We define, in the spirit of Kretschmer \cite{Kretschmer21Quantum}, a query to a single unitary $\U$ to be a single quantum call of either $\U$ or controlled-$\U$. We do not allow queries to $\U^{\dagger}$. $\A^\U(x)$ refers to a quantum algorithm on a classical input $x$ which can make quantum queries to the unitary (or collection of unitaries) $\U$. In terms of computational cost, a single query to $\U_n$ will be charged $n$ units of computation. This allows us to define quantum polynomial-time (QPT) algorithms relative to an oracle $\U$. In particular, a QPT algorithm relative to $\U$ on an input of length $\ell$ can query $\U_n$ for any $n < \poly(\ell)$.

Also in the style of Kretschmer, we consider versions of $\PromBQP$, $\mathsf{Promise}$ $\mathsf{QCMA}$, and $\PromQMA$ augmented with a collection of quantum oracles $\U = \{\U_n\}_{n \in \N}$. We denote these by $\PromBQP^\U$, $\PromQCMA^\U,$ and $\PromQMA^\U$ respectively. For $\PromBQP^\U$, the deciding algorithm is allowed to be a QPT algorithm relative to $\U$, and for $\PromQCMA^\U$ and $\PromQMA^\U$, the verifying algorithm is allowed to be a QPT algorithm relative to $\U$. 

It is easy to see that in this model, the traditional inequalities still hold. In particular, for any oracle $\U$, $\PromBQP^\U \subseteq \PromQCMA^\U \subseteq \PromQMA^\U$.

We also consider cryptographic primitives in the oracle setting. In this case, we allow the cryptographic algorithm to be a uniform QPT algorithm relative to $\U$, and we consider security against non-uniform QPT algorithms relative to $\U$. 


\section{Constructions of $\EVOWP$ from QCCC primitives}

In this section we give the results that $\EVOWP$ are equivalent to QCCC one time signatures, and can be constructed from QCCC non-interactive commitments and $\QPRG$s. 

\begin{theorem}\label{thm:otsowp}
    There exists a one-time signature scheme if and only if there exists a $\EVOWP$.
\end{theorem}

\begin{proof}[Proof of \cref{thm:otsowp}]
    \cite{khurana2024commitments} show that you can construct $\OWP$ from one-time signature schemes (in fact zero-time signature schemes). They do not define $\EVOWP$, but it is clear that their construction has efficient verification. We repeat their construction here for completeness.

    Let $(KeyGen,S,V)$ be a one-time signature scheme. Then $(\Samp,\Ver)$ defined as follows is a $\EVOWP$
    \begin{enumerate}
        \item $\Samp(1^\lambda)$: Sample $KeyGen(1^\lambda) \to (vk,sk)$. Output $(k=vk,s=sk)$
        \item $\Ver(k,s)$: Sample $m$ uniformly at random from the message space. Output $V(s,m,S(k,m))$.
    \end{enumerate}

    To show the other direction (that $\EVOWP \rightarrow$ signitures), it is not hard to see that the Lamport signature scheme \cite{lamport1979constructing} building one-time signatures from one-way functions can be generalized to work with $\EVOWP$. In particular, let $(\Samp,\Ver)$ be a $\EVOWP$, we define a signature scheme using it as follows. For simplicity, the message space will be $\{0,1\}$.
    \begin{enumerate}
        \item $KeyGen(1^\lambda)$: Run $\Samp$ twice to generate two key-puzzle pairs $(k_0,s_0)$ and $(k_1,s_1)$. Output $(vk = (s_0,s_1), sk = (k_0,k_1))$.
        \item $S((k_0,k_1), b)$: Output $k_b$.
        \item $V((s_0,s_1), b, \sigma)$: Output $\Ver(\sigma, s_b)$.
    \end{enumerate}

    Correctness is immediate from correctness of the one-way-puzzle scheme. To show security, we will assume towards contradiction that their exists some pair of messages $m_0\neq m_1$ and an adversary $\A$ breaking security of the signature scheme. Without loss of generality we will assume $m_0 = 0$ and $m_1 = 1$. Thus,
    $$\Pr[V(vk,1,\A(vk,S(sk,0))) \to 1] > \lambda^{-c}$$
    for some $c$. Rewriting this in the notation of the underlying one way puzzle we have
    $$\Pr_{\Samp(1^\lambda) \to (k_0,s_0),(k_1,s_1)}[\Ver(\A((s_0,s_1),k_0), s_1)] > \lambda^{-c}$$
    We will define a new adversary $\mathcal{B}$ breaking the one-way puzzle as follows. On input $s$, $\mathcal{B}$ runs $\Samp(1^{\lambda}) \to (k',s')$ and outputs $\A((s,s'),k')$. It is clear that
    $$\Pr_{\Samp(1^\lambda)\to (k,s)}[\Ver(\mathcal{B}(s),s) \to 1]$$ $$ =  \Pr_{\Samp(1^\lambda) \to (k_0,s_0),(k_1,s_1)}[\Ver(\A((s_0,s_1),k_0), s_1)] > \lambda^{-c}$$
    But as $(\Samp,\Ver)$ is a $\EVOWP$, this is a contradiction, and so the one-time signature scheme is secure.
\end{proof}

\cite{khurana2024commitments} shows that $\EVOWP$ can be built from one-time signatures even if the signing key or the signature are quantum. Thus, an interesting corollary of Theorem~\ref{thm:otsowp} is that QCCC one-time signatures with classical signature, signing and verification keys can be built from one-time signatures where either the signing key or the signature is quantum.

\begin{theorem}\label{thm:nicom}
    If there exists a non-interactive commitment scheme $(Com,Rec)$, then there exists a $\EVOWP$ $(\Samp,\Ver)$.
\end{theorem}

\begin{proof}[Proof of \cref{thm:nicom}]
    Our construction is as follows
    \begin{enumerate}
        \item $\Samp$: Pick $m$ uniformly at random. Run $Com(m)\to (c,d)$. Output $(k=(m,d),s=c)$.
        \item $\Ver(k=(m,d),s=c)$: Run $Rec(c,d) \to m'$. Output $1$ if and only if $m' = m$. 
    \end{enumerate}
    Correctness immediately implies that
    $$\Pr_{\Samp\to (k,s)}[\Ver(k,s)] = \Pr_{\$\to m,Com(m)\to (c,d)}[Rec(c,d) = m] \geq 1 -\negl(\lambda)$$

    We now proceed to show security. Let $\A$ be any QPT adversary. We will show
    $$\Pr_{\Samp \to (k,s)}[\Ver(\A(s),s)\to 1] \leq \negl(\lambda)$$
    Observe that
    \begin{equation*}
        \begin{split}
            \Pr_{\Samp \to (k,s)}[\Ver(\A(s),s)\to 1] \\
            = \Pr_{Com(m)\to (c,d)}[Rec(c,d') = m' ; \A(c) \to (d',m')]\\
            = \Pr_{Com(m)\to (c,d)}[Rec(c,d') = m' \wedge m = m' ; \A(c) \to (d',m')] \\
            + \Pr_{Com(m)\to (c,d)}[Rec(c,d') = m' \wedge m \neq m' ; \A(c) \to (d',m')]
        \end{split}
    \end{equation*}
    But hiding implies that the probability that $\A$ computes $m$ from $c$ is negligible, so 
    $$\Pr_{Com(m)\to (c,d)}[Rec(c,d') = m' \wedge m = m' ; \A(c) \to (d',m')] \leq \negl(\lambda)$$
    And binding says that after an honest commitment, there is no way to open to a different message, and so 
    $$\Pr_{Com(m)\to (c,d)}[Rec(c,d') = m' \wedge m \neq m' ; \A(c) \to (d',m')] \leq \negl(\lambda)$$
    Together, we have
    $$\Pr_{\Samp \to (k,s)}[\Ver(\A(s),s)\to 1] \leq \negl(\lambda)$$
\end{proof}

\begin{theorem}\label{thm:qprgtoowp}
    If there exists a $\QPRG$ $G$ with stretch $\ell(\lambda) \geq 3n(\lambda)$, then there exists a $\EVOWP$ $(\Samp,\Ver)$.
\end{theorem}

\begin{proof}[Proof of \cref{thm:qprgtoowp}]
    We define our $\EVOWP$ as follows
    \begin{enumerate}
        \item $\Samp(1^{\lambda})$: Sample $k_1,\dots, k_{\lambda}$ uniformly from $\{0,1\}^n$. Sample $G_{\lambda}(k_i) \to s_i$. Output key $k = k_1,\dots,k_{\lambda}$, and puzzle $s = s_1,\dots, s_{\lambda}$.
        \item $\Ver(k,s)$: For each $i$ sample $G(k_i) \to \wt{s_i}$. 
        If for any $i$, $s_i = \wt{s_i}$, output $1$. 
        Otherwise, output $0$.
    \end{enumerate}

    We first will show correctness. Pseudodeterminism of $G$ gives us that for sufficiently large $\lambda$, the probability over a random $k$ that two runs of $G(k)$ give the same result is $\geq \frac{1}{2}$. Thus,
    \begin{equation*}
        \begin{split}
            \Pr_{\Samp(1^{\lambda})\to(k,s)}[\Ver(k,s)\to 0] &= \prod_{i=1}^\lambda \Pr_{\{0,1\}^n \to k_i, G(k_i) \to s_i,G(k_1) \to \wt{s_i}}[s_i = \wt{s_i}]\\
            &\leq \prod_{i=1}^\lambda \frac{1}{2} \leq \frac{1}{2^\lambda} \leq \negl(\lambda)
        \end{split}
    \end{equation*}
    which gives us correctness.

    We will now argue security via a reduction. Let $\A$ be any non-uniform QPT algorithm such that
    $$\Pr_{\Samp(1^\lambda)\to (k,s)}[\Ver(\A(s),s) \to 1] \geq \epsilon$$
    By an averaging argument, there must exist some index $i\in[\lambda]$ such that
    $$\Pr_{\Samp(1^\lambda) \to (k,s)}[G((\A(s))_i) = s_i] \geq \frac{\epsilon}{\lambda}$$

    We define an adversary $\A'$ against $G$ as follows
    \begin{enumerate}
        \item On input $y$
        \item Sample $k_1,\dots,k_{i-1},k_{i+1},\dots,k_{\lambda}$ uniformly from $\{0,1\}^n$
        \item Send $G(k_1),\dots,G(k_{i-1}),y,G(k_{i+1}),\dots,G(k_{\lambda})$ to $\A$, getting response $\wt{k}$
        \item If $G(\wt{k}_i) = y$, then output $1$, otherwise output $0$,
    \end{enumerate}

    First, we will lower bound $\Pr_{\{0,1\}^n \to x}[\A'(G(x)) \to 1]$. Note that on input $y = G(x)$, the view of $\A$ is exactly as it should be in its own game, and so
    $$\Pr_{\{0,1\}^n \to x}[\A'(G(x)) \to 1] = \Pr_{\Samp(1^\lambda) \to (k,s)}[G((\A(s))_i) = s_i] \geq \frac{\epsilon}{\lambda}$$

    Now, we will upper bound $\Pr_{\{0,1\}^\ell \to y}[\A'(G(y)) \to 1]$. Note that the optimal adversary for this problem, on input $y$, returns $\argmax_x \Pr[G(x) = y]$. Thus, we have the following trivial bound
    $$\Pr_{\{0,1\}^\ell \to y}[\A'(G(y)) \to 1] \leq \E_{\{0,1\}^\ell \to y}[\max_x \Pr[G(x) = y]]$$
    But there can be at most $2^{2n}$ values of $y$ such that $\max_x \Pr[G(x) = y] \geq \frac{1}{2^n}$. To see this, observe
    \begin{equation*}
        \begin{split}
            \sum_y \max_x \Pr[G(x) = y] \leq \sum_{x,y} \Pr[G(x) = y] = \sum_x \sum_y \Pr[G(x) = y] = \sum_x 1 = 2^n
        \end{split}
    \end{equation*}
    But since $\ell(\lambda) \geq 3n(\lambda)$, we have that
    $$\Pr_{\{0,1\}^\ell \to y}[\max_x \Pr[G(x) = y] \geq 2^{-n}] \leq \frac{2^{2n}}{2^{\ell}} \leq \frac{2^{2n}}{2^{3n}} = 2^{-n}$$
    Thus, we conclude with
    \begin{equation*}
    \begin{split}
        \E_{\{0,1\}^\ell \to y}[\max_x \Pr[G(x) = y]] &\leq \Pr_{\{0,1\}^\ell \to y}[\max_x \Pr[G(x) = y] \geq 2^{-n}]\cdot 1 + 2^{-n}\\
        &\leq 2^{-n} + 2^{-n}
    \end{split}
    \end{equation*}

    Putting this together, we get that
    $$\Pr_{\{0,1\}^n \to x}[\A'(G(x)) \to 1] - \Pr_{\{0,1\}^\ell \to y}[\A'(y) \to 1] \geq \frac{\epsilon}{\lambda} + 2^{-n+1}$$
    and so since
    $$\Pr_{\{0,1\}^n \to x}[\A'(G(x)) \to 1] - \Pr_{\{0,1\}^\ell \to y}[\A'(y) \to 1] \leq \negl(\lambda)$$
    by $\QPRG$ security, we get that $\epsilon \leq \negl(\lambda)$ and so $(\Samp,\Ver)$ satisfies security.
\end{proof} 

\section{Efficiently verifiable one way puzzles can be broken with a $\QCMA$ oracle}

\begin{proposition}[From~\cite{synthesis,Aharonov_2022}]\label{prop:searchtodecision}
    There exists a search-to-decision reduction for $\PromQCMA$. Formally, there exists a promise problem $\Pi^* \in \PromQCMA$ such that for every $\Pi \in \PromQCMA$ with verifier $\mathcal{V}$, there exists a QPT algorithm $\mathcal{A}^{\Pi^*}$ such that for all $x$ such that $\Pi(x) = 1$,
    $$\Pr[\Pr[\mathcal{V}(\A^{\Pi}(x),x) \to 1]\geq \frac{2}{3}] \geq \frac{1}{2}$$
\end{proposition}

\begin{theorem}\label{thm:qcmaoracle}
    For every efficiently verifiable one way puzzle $(\Samp,\Ver)$, there exists a promise problem $\Pi^*\in \PromQCMA$ and a QPT algorithm $\A^{\Pi^*}$ with oracle access to $\Pi^*$ which breaks security. That is
    $$\Pr_{\Samp(1^\lambda)\to (k,s)}[\Ver(\A^{\Pi^*}(s),s)\to 1] \geq \frac{1}{\poly(\lambda)}$$
\end{theorem}

\begin{proof}
    Let $(\Samp,\Ver)$ be a $\EVOWP$. Define $\Pi$ to be the following promise problem:
    \begin{enumerate}
        \item (yes): $s$ is a yes instance if there exists $k$ such that $\Pr[\Ver(k, s) \to 1]\geq \frac{2}{3}$
        \item (no): $s$ is a no instance if for all $k$, $\Pr[\Ver(k, s) \to 1] < \frac{1}{3}$
    \end{enumerate}
    It is trivial to see that $\Pi \in \PromQCMA$.

    Our promise problem $\Pi^*$ will be the problem from~\Cref{prop:searchtodecision} and our algorithm $\A^{\Pi^*}$ will be the 
    algorithm from~\Cref{prop:searchtodecision} corresponding to $\Pi$. In particular, we have that for all $s$ such that there exists a $k$ 
    with $\Pr[\Ver(k,s) \to 1] \geq \frac{2}{3}$, then
    $$\Pr[\Pr[\Ver(\A^{\Pi^*}(s),s) \to 1] \geq \frac{2}{3}] \geq \frac{1}{2}$$
    In particular, this means that
    $$\Pr[\Ver(\A^{\Pi^*}(s),s) \to 1] \geq \frac{1}{3}$$

    Correctness states that $$\Pr_{\Samp(1^\lambda) \to (k,s)}[\Ver(k,s) \to 1] \geq 1 - \negl(\lambda).$$
    And so an averaging argument gives us that
    $$\Pr_{\Samp(1^\lambda) \to (k,s)}[\Pr[\Ver(k,s) \to 1] \geq \frac{2}{3}] \geq 1 - \negl(\lambda) \geq \frac{1}{2}$$
    In particular, $s$ is a yes instance of $\Pi$ with all but negligible probability.

    Putting things together, we get that
    $$\Pr_{\Samp(1^\lambda) \to (k,s)}[\Ver(\A^{\Pi^*}(s),s) \to 1] \geq \frac{1}{6}$$
    and we are done.
\end{proof} 

\section{A black-box separation between $\OWP$ and $\EVOWP$}

We begin by recalling the very powerful quantum black-box separation theorem by Kretschmer.

\begin{theorem}[\cite{Kretschmer21Quantum}]\label{thm:kreschmer}
    There exists a set of quantum oracles $\U$ such that with probability $1$ over $\U$,
    \begin{enumerate}
        \item $\PromBQP^{\U} = \PromQMA^\U$.
        \item Relative to $\U$, there exists a PRS family mapping $\lambda$ bits to $\lambda$ qubit states.
    \end{enumerate}
\end{theorem}

First, we observe that all of our theorems are black-box and thus relativize. In particular, we get the following corollaries

\begin{corollary}\label{cor:qcma}
    Let $\U$ be any collection of classical or quantum oracles. For every efficiently verifiable one way puzzle $(\Samp^\U,\Ver^\U)$, there exists a promise problem $\Pi\in \PromQCMA^\U$ and a QPT algorithm $\A^\Pi$ with oracle access to $\Pi$ which breaks security. That is
    $$\Pr_{\Samp^\U(1^\lambda)\to (k,s)}[\Ver^\U(\A^\Pi(s),s)\to 1] \geq \frac{1}{\poly(\lambda)}$$
\end{corollary}

\begin{corollary}\label{cor:qprg}
    Let $\U$ be any collection of classical or quantum oracles. If there exists a $\QPRG$ $G^\U$ with stretch $\ell(\lambda) \geq 3n(\lambda)$ secure relative to $\U$, then there exists a $\EVOWP$ $(\Samp^\U,\Ver^\U)$ secure relative to $\U$.
\end{corollary}

Let us then consider the oracle $\U$ from Theorem~\ref{thm:kreschmer}. We know that relative to $\PromBQP^\U \subseteq \PromQCMA^\U \subseteq \PromQMA^\U$, and we also have $\PromBQP^\U = \PromQMA^\U$. Thus, this gives us that $\PromBQP^\U = \mathsf{Promise}$ $\mathsf{QCMA}^\U$. So Corollary~\ref{cor:qcma} immediately shows that relative to $\U$, there does not exist any $\EVOWP$.

But it is also known from \cite{khurana2024commitments} and \cite{morimae2022quantum} that $\OWP$ can be built in a black-box manner from $\PRS$s. Thus, we get the following corollary:

\begin{theorem}\label{thm:separation}
    There exists a set of quantum oracles $\U$ such that with probability $1$ over $\U$,
    \begin{enumerate}
        \item There does not exist any $\EVOWP$ relative to $\U$.
        \item There exists a $\OWP$ relative to $\U$.
    \end{enumerate}
\end{theorem}

~\cite{ALY23} shows that PRSs with output length $c \log \lambda$ for $c > 12$ can be used to build $\QPRG$s with triple stretch. Note that this reduction is itself black-box, and so holds relative to any quantum oracle. Applying Corollary~\ref{cor:qprg} then gives the following result

\begin{corollary}\label{cor:prsseparation}
    There exists a set of quantum oracles $\U$ such that with probability $1$ over $\U$, relative to $\U$
    \begin{enumerate}
        \item There does not exist any PRS family mapping $\lambda$-bits to $c \log \lambda$-qubits for any constant $c > 12$.
        \item There does exist a PRS family mapping $\lambda$-bits to $\lambda$-qubits.
    \end{enumerate}
\end{corollary} 

\section{ $\OWP$ Amplification, Combiners, and Universal Constructions}

In this section we present security and correctness amplifiers, combiners, and universal constructions for both $\OWP$ and $\EVOWP$. 
We defer the proof of $\OWP$ security amplification to the 
next section. With these results we establish that both primitives are well behaved and have many of the desirable properties of one-way 
functions.

\subsection{Amplification}

\begin{theorem}[Correctness amplification for $\OWP$ and $\EVOWP$]\label{thm:owpcorramp}
    Let $(\Samp,\Ver)$ be a $(\alpha, \beta)$ $\OWP$. Define $(\Samp',\Ver')$ by
    \begin{enumerate}
        \item $\Samp' = \Samp^{\otimes t}$
        \item $\Ver'((k_1,\dots,k_t),(s_1,\dots,s_t))$: output $1$ if $\Ver(k_i,s_i)$ for some $i\in[t]$.
    \end{enumerate}
    Then $(\Samp',\Ver')$ is a $(\alpha^t, t\beta) \OWP$.

\end{theorem}

\begin{proof}[Proof of \cref{thm:owpcorramp}]
    To see correctness, we observe that
    \begin{equation*}
        \begin{split}
            \Pr_{\Samp'(1^\lambda)\to((k_1,\dots,k_t),(s_1,\dots,s_t))}[\Ver'((k_1,\dots,k_t),(s_1,\dots,s_t))=0]\\ =\Pr_{\Samp'(1^\lambda)\to((k_1,\dots,k_t),(s_1,\dots,s_t))}[\text{ all of }\Ver(k_i,s_i) = 0]\\
            \leq \Pr_{\Samp(1^\lambda)\to(k,s)}[\Ver(k,s) = 0] \leq \alpha^t
        \end{split}
    \end{equation*}

    To show security, we will do a simple reduction. Let $\A$ be such that
    $$\Pr_{\Samp'(1^\lambda)\to ((k_1,\dots,k_t),(s_1,\dots,s_t))}[\Ver(\A(s_1,\dots,s_t),(s_1,\dots,s_t))\to 1] > t\beta$$
    We will construct an adversary $\mathcal{B}$ breaking $(\Samp,\Ver)$.
    \begin{enumerate}
        \item On input $s$, pick $i$ uniformly at random from $[t]$. Set $s_i = s$.
        \item Run $\Samp(1^\lambda)$ $t-1$ times to generate $s_j$ for $j\neq i$.
        \item Output $\A(s_1,\dots,s_t)$.
    \end{enumerate}
    It is clear that
    \begin{equation*}
    \begin{split}
        \Pr_{\Samp(1^\lambda)\to (k,s)}[\Ver(\mathcal{B}(s),s)\to 1]\\
        =\Pr_{\Samp'(1^\lambda)\to ((k_1,\dots,k_t),(s_1,\dots,s_t)),[t] \to i}[\Ver(\mathcal{A}(s_1,\dots,s_t)_i,s_i)]\\
        \geq \frac{1}{t}\Pr_{\Samp'(1^\lambda)\to ((k_1,\dots,k_t),(s_1,\dots,s_t))}[\exists i\text{ s.t. }\Ver(\A(s_1,\dots,s_t)_i,s_i)\to 1]\\
        \geq \frac{1}{t}\Pr_{\Samp'(1^\lambda)\to ((k_1,\dots,k_t),(s_1,\dots,s_t))}[\Ver'(\A(s_1,\dots,s_t),(s_1,\dots,s_t))\to 1] > \frac{1}{t}t\beta > \beta
        \end{split}
    \end{equation*}
    Thus, $(\Samp',\Ver')$ satisfies $\beta$ security and we are done.
\end{proof}

\begin{remark}
    We note that if $\Ver$ is efficient, so is $\Ver'$. 
\end{remark}



\begin{theorem}[Weak correctness amplification for $\OWP$]\label{thm:weakowpcorramp}
    Let $(\Samp,$ $\Ver)$ be a $(\alpha,\beta)$ $\OWP$. Define $(\Samp',$ $\Ver')$ by
    \begin{enumerate}
        \item $\Samp'=\Samp$
        \item $\Ver'(k,s)$: If $\Pr_{\Samp(1^\lambda)\to (k',s')}[\Ver(k',s') \to 0 | s' = s] \geq t$, output $1$. Otherwise, output $\Ver(k,s)$. (The idea here is that if $s$ was sampled in a way that it would not verify honestly then we accept it anyway, otherwise we do normal verification)
    \end{enumerate}
    Then $(\Samp',\Ver')$ is a $(t,\alpha/t + \beta)$ $\OWP$.
\end{theorem}

\begin{proof}[Proof of \cref{thm:weakowpcorramp}]
    To show correctness, we observe that 
    \begin{equation*}
        \begin{split}
            \Pr_{\Samp'(1^\lambda)\to (k,s)}[\Ver'(k,s) \to 0]\\
            =\Pr_{\Samp(1^\lambda)\to (k,s)}\left[\Ver(k,s) \to 0 \wedge \Pr_{\Samp(1^\lambda)\to (k',s')}[\Ver(k',s') \to 0 | s' = s] < t\right]\\
            \leq \Pr_{\Samp(1^\lambda)\to (k,s)}\left[\Ver(k,s) \to 0 | \Pr_{\Samp(1^\lambda)\to (k',s')}[\Ver(k',s') \to 0 | s' = s] < t\right]\\
            < t
        \end{split}
    \end{equation*}

    Before showing security, let us define $X$ to be the random variable defined by sampling $\Samp(1^\lambda) \to (k,s)$ and outputting $\Pr_{\Samp(1^\lambda)\to (k',s')}[\Ver(k',s') \to 0 | s' = s]$. We will begin by investigating $\Pr[X < t]$.

    Observe that $\E[X] = \Pr_{\Samp(1^\lambda)\to (k,s)}[\Ver(k,s) \to 0] \leq \alpha$. Markov's inequality then gives that
    $$\Pr[X \geq t] \leq \frac{\alpha}{t}.$$

    Let $\A$ be any PPT algorithm
    \begin{equation*}
        \begin{split}
            \Pr_{\Samp'(1^\lambda)\to (k,s)}[\Ver'(\A(s),s) \to 1]\\
            =\Pr_{\Samp'(1^\lambda)\to(k,s)}[\Ver(\A(s),s) \to 1 \vee X \geq t]\\
            \leq \Pr_{\Samp'(1^\lambda) \to (k,s)}[\Ver(\A(s),s)\to 1] + \Pr[X \geq t]\\
            \leq \beta + \frac{\alpha}{t}
        \end{split}
    \end{equation*}
\end{proof}

\begin{remark}
    Note that this construction only holds for standard inefficiently verifiable one-way puzzles, as $\Ver'$ is not efficient. 
\end{remark}

\begin{theorem}[Security amplification for OWP]\label{thm:owpamplification}
    If, for some $c>0$, there exists a $(\negl(\lambda), 1-\lambda^{-c})$ one-way puzzle $(\Samp,\Ver)$, then there exists a strong one-way puzzle.
\end{theorem}

We defer the proof of the above theorem to the next section.

\begin{theorem}[Weak correctness amplification for $\EVOWP$]\label{thm:weakevowpcorr}
    Let $(\Samp,\Ver)$ be a $(\alpha,\beta)$ $\EVOWP$. Define $(\Samp',\Ver')$ by
    \begin{enumerate}
        \item $\Samp'(1^\lambda)$: Run $\Samp(1^\lambda) \to (k,s)$. If $\Ver(k,s) \to 1$, output $(k,s)$. Otherwise, output $(\bot,\bot)$.
        \item $\Ver'(k,s)$: If $s = \bot$, output $1$. Otherwise, output $\Ver(k,s)$.
    \end{enumerate}
    Then $(\Samp',\Ver')$ is a $(1/4,\alpha+\beta)$ $\OWP$.
\end{theorem}

\begin{theorem}[Security amplification for $\EVOWP$]\label{thm:secampevowp}
    Let $(\Samp,\Ver)$ be a $(\alpha,\beta)$ $\EVOWP$. Then $(\Samp^{\otimes t},\Ver^{\otimes t})$ is a $(t\alpha,\beta^t)$ $\EVOWP$.
\end{theorem}
\begin{proof}[Proof of \cref{thm:secampevowp}]
    Correctness follows from the union bound applied to correctness of $(\Samp,\Ver)$:
    \begin{equation*}
    \begin{split}
        \Pr_{\Samp^{\otimes t}(1^\lambda)\to ((k_1,\dots,k_t),(s_1,\dots,s_t))}[\Ver^{\otimes t}((k_1,\dots,k_t),(s_1,\dots,s_t)) \to 0]\\
        = \Pr_{\Samp^{\otimes t}(1^\lambda)\to ((k_1,\dots,k_t),(s_1,\dots,s_t))}[\text{ there exists an }i\text{ such that }\Ver(k_i,s_i) \to 0]\\
        \leq \sum_i\Pr_{\Samp(1^\lambda)\to(k_i,s_i)}[\Ver(k_i,s_i)\to 0]\\
        \leq t\alpha
    \end{split}
    \end{equation*}

    Security follows from Theorem 4.1 from~\cite{bostanci2023efficient}. 
\end{proof}

\begin{remark}
    The proof of Theorem 4.1 from~\cite{bostanci2023efficient} requires that the soundness game be efficiently falsifiable. Thus, this same amplification theorem does not hold for $\OWP$ with inefficient verification. Amplification of such puzzles is an open question.
\end{remark}



\subsection{Combiners}

\begin{corollary}\label{cor:owpcorrect}
    Let $(\Samp,\Ver)$ be a $\OWP$ candidate and define $(\Samp',\Ver')$ to be the constructions from~\Cref{thm:weakowpcorramp} and~\Cref{thm:owpcorramp} applied in sequence with $t=1/2$ and $t=\lambda$ respectively. Then if $(\Samp,\Ver)$ is a $\OWP$, so is $(\Samp',\Ver')$. Furthermore, regardless of whether $(\Samp,\Ver)$ is a $\OWP$, $(\Samp',\Ver')$ satisfies $n^{-c}$ correctness for all $c$.
\end{corollary}

\begin{proof}
    Let us say that $(\Samp,\Ver)$ satisfies $(\alpha,\beta)$ correctness and security. Then $(\Samp',\Ver')$ satisfies $(2^{-\lambda}, \lambda(\alpha/2 + \beta))$ correctness and security.

    Observe that if $\alpha,\beta = \negl(\lambda)$, then $\lambda(\alpha/2 + \beta) = \negl(\lambda)$. Furthermore, no matter what, $2^{-\lambda} = \negl(\lambda)$.
\end{proof}

\begin{corollary}\label{cor:evowpcorrect}
    Let $(\Samp,\Ver)$ be a $\EVOWP$ candidate and define $(\Samp',$ $\Ver')$ to be the constructions from~\Cref{thm:weakevowpcorr} and~\Cref{thm:owpcorramp} applied in sequence with $t=\lambda$. Then if $(\Samp,\Ver)$ is a $\EVOWP$, so is $(\Samp',\Ver')$. Furthermore, regardless of whether $(\Samp,\Ver)$ is a $\EVOWP$, $(\Samp',\Ver')$ satisfies $n^{-c}$ correctness for all $c$.
\end{corollary}

\begin{proof}[Proof of \cref{thm:weakevowpcorr}]
    By correctness, $\Pr[\Samp'(1^\lambda)\to (\bot,\bot)] \leq \alpha$. Thus, 
    \begin{equation*}
        \begin{split}
            \Pr_{\Samp'(1^\lambda)\to (k,s)}[\Ver'(k,s)\to 0]\\
            = \Pr_{\Samp'(1^\lambda)\to (k,s)}[\Ver'(k,s)\to 0 \wedge s = \bot]\\
            + \Pr_{\Samp'(1^\lambda)\to (k,s)}[\Ver'(k,s)\to 0 \wedge s\neq \bot]\\
            = \Pr_{\Samp(1^\lambda) \to (k,s)}[\Ver(k,s) \to 0 \wedge \Ver(k,s)\to 1]\\ \text{ where these are two separate calls to }\Ver
        \end{split}
    \end{equation*}
    Take any $k,s$. Say $p \coloneqq \Pr[\Ver(k,s) \to 1]$. Then $\Pr[\Ver(k,s)\to 0\wedge \Ver(k,s)\to 1] = p(1-p) \leq \frac{1}{4}$. Thus, 
    $$\Pr_{\Samp'(1^\lambda)\to (k,s)}[\Ver'(k,s)\to 0] \leq 1/4$$

    Let $\A$ be any PPT algorithm. We will prove security by showing that
    $$\Pr_{\Samp'(1^\lambda)\to (k,s)}[\Ver'(\A(s),s) \to 1] \leq \beta + \alpha$$
    \begin{equation*}
    \begin{split}
        \Pr_{\Samp'(1^\lambda)\to (k,s)}[\Ver'(\A(s),s) \to 1]\\
        = \Pr_{\Samp'(1^\lambda)\to (k,s)}[\Ver(\A(s),s) \to 1 \vee s = \bot]\\
        \leq \Pr_{\Samp'(1^\lambda)\to (k,s)}[\Ver(\A(s),s) \to 1] + \Pr_{\Samp'(1^\lambda)\to (k,s)}[s = \bot]\\
        = \alpha + \beta
        \end{split}
    \end{equation*}
    by the union bound.
\end{proof}

\begin{proof}
    Same argument as before
\end{proof}

\begin{corollary}\label{cor:owpcomb}
    Let $(\Samp_0,\Ver_0)$ and $(\Samp_1,\Ver_1)$ be two $\OWP$ candidates. Let $(\Samp_0',\Ver_0')$ and $(\Samp_1',\Ver_1')$ be the construction from~\Cref{cor:owpcorrect} applied to $(\Samp_0,\Ver_0)$ and $(\Samp_1,\Ver_1)$ respectively. Define $(\wt{\Samp},\wt{\Ver})$ by
    \begin{enumerate}
        \item $\wt{\Samp}(1^\lambda)$: Run $\Samp_0' \to (k_0,s_0)$ and $\Samp_1'\to (k_1,s_1)$. Output $(k=(k_0,k_1),s=(s_0,s_1))$.
        \item $\wt{\Ver}((k_0,k_1),(s_0,s_1))$: Output $1$ if and only if $\Ver_0'(k_0,s_0)=1$ and $\Ver_1'(k_1,s_1)=1$.
    \end{enumerate}
    Then $(\wt{\Samp},\wt{\Ver})$ is a $\OWP$ as long as one of $(\Samp_0,\Ver_0)$ or $(\Samp_1,\Ver_1)$ is a $\OWP$.

    The same corollary holds for $\EVOWP$ when replacing~\Cref{cor:owpcorrect} with~\Cref{cor:evowpcorrect}.
\end{corollary}

\begin{remark}\label{cor:owpcombmany}
    Given any polynomial $p$ different $\OWP$ candidates $\{(\Samp_i,\Ver_i)\}_{i\in[p]}$, we can use the same idea to build a $\OWP$ $(\wt{\Samp},\wt{\Ver})$ which is secure as long as one of the candidates is secure. This follows the same logic as above with $\wt\Samp(1^\lambda) \rightarrow ((k_0,...,k_p),(s_0,...,s_p))$ and $\Ver$ only outputting $1$ if all $p$ $\Ver'_i$'s output $1$.
\end{remark}

\subsection{Universal Construction}

The universal construction of $\EV$ follows standard techniques.

\begin{definition}[Universal $\EVOWP$]
    A set of uniform algorithms $(\Samp,$ $\Ver)$ is a universal construction of a $\EVOWP$ if it is a $\EVOWP$ so long as some $\EVOWP$ exists. 
\end{definition}

\begin{theorem}[Universal $\EVOWP$]\label{thm:universalevowp}
    Define $\Samp_i$ as the $i^{th}$ Quantum Turing machine and $\Ver_j$ as the $j^{th}$ Quantum Turing machine each of which with an attached alarm clock which will halt the machine after $\lambda^3$ steps. 

    Define $(\Samp_U,\Ver_U)$ using the construction from \Cref{cor:owpcombmany} for the $\lambda^2$ candidates $(\Samp_i,\Ver_j)$ for all $(i,j) \in [\lambda]^2$. If any $\EVOWP$ $(\Samp,\Ver)$ exists, then $(\Samp_U,\Ver_U)$ is a $\OWP$.
\end{theorem}
\begin{proof}
    Using standard padding arguments (see \cite{Levin87,HKNY23}) we can show that if there exists a $(\alpha,\beta)$ $  \EVOWP$ with both $\Ver$ and $\Samp$ running in time less than $\lambda^c$, then there exists a $(\alpha^{(1/c)},\beta^{(1/c)})$ $\EVOWP$ with $\Ver'$ and $\Samp'$ running in time less than $\lambda^3$. If $\EVOWP$'s exist then there exists a $(\negl,\negl)\EVOWP$ and therefore a $(\negl,\negl)\EVOWP$ running in time less than $\lambda^3$.
    
    If any one-way puzzle exists then for some $i,j$, $(\Samp_i,\Ver_j)$ will be a one way puzzle with both algorithms running in time less than $n^3$. Once $\lambda \geq \max(i,j)$ a one way puzzle will be one of the candidates. By \Cref{cor:owpcombmany} we know that $(\Samp_U,\Ver_U)$ will be a $\OWP$.
\end{proof} 

\begin{theorem}[Universal $\OWP$]\label{thm:uowp}
    Define $\Samp_i$ as the $i^{th}$ Quantum Turing machine with an attached alarm clock which will halt the machine after $\lambda^3$ steps. Define $\Ver_i$ to be the function such that the sum of the correctness and security error of $(\Samp_i,\Ver_i)$ is minimized.

    Define $(\Samp_U)(1^\lambda)$ using the construction from~\Cref{cor:owpcombmany} for the $\lambda^2$ candidates $(\Samp_i,\Ver_i)$ for all $(i,j) \in [\lambda]^2$. If any $\OWP$ $(\Samp,\Ver)$ exists, then $(\Samp_U,\Ver_U)$ is a $\OWP$.
\end{theorem}

\begin{proof}
    Observe that if there exists a $\OWP$ $(\Samp,\Ver)$ then when $i$ is such that $\Samp_i = \Samp$, $(\Samp_i,\Ver_i)$ is also a $\OWP$. The rest of the proof follows by the same argument as~\Cref{thm:universalevowp}.
\end{proof}

\section{$\OWP$ security amplification}

In this section we prove the following theorem:

\begin{theorem}[Restatement of~\Cref{thm:owpamplification}]
    If, for some $c>0$, there exists a $(\negl(\lambda), 1-\lambda^{-c})$ one-way puzzle $(\Samp,\Ver)$, then there exists a strong one-way puzzle.
\end{theorem}

We do this be showing that weak $\OWP$ imply non-uniform EFID pairs, and then show that non-uniform EFID pairs imply strong $\OWP$. The first of these steps serves as a simpler and more general version of the argument presented in \cite{khurana2024commitments}.

\subsection{$\OWP$s imply non-uniform EFID }

\begin{definition}[From~\cite{khurana2024commitments}]
    A $\nu^*$-non-uniform EFID pair is a QPT algorithm $G_\nu(1^\lambda,b)$ with classical parameter-dependent advice $\nu$. On input a unary security parameter $\lambda$ and $b\in\{0,1\}$ outputs a classical string such that
    \begin{enumerate}
        \item For all QPT $\A$ and for all sufficiently large $\lambda$, the distributions $G_{\nu^*}(1^\lambda,0)$ and $G_{\nu^*}(1^\lambda, 1)$ are indistinguishable.
        \item $$\Delta(G_{\nu^*}(1^\lambda, 0), G_{\nu^*}(1^\lambda, 1)) \geq 1 - \epsilon$$
    \end{enumerate}
\end{definition}

\begin{theorem}[Restatement of~\Cref{cor:nbpetoqefid}]
    If there exists a $\OWP$, then there exists a non-uniform EFID pair.
\end{theorem}

\begin{corollary}[Restatement of~\Cref{cor:weakowptoefi}]
    If there exists a $\OWP$, then there exists an EFI pair.
\end{corollary}

This argument will follow from adapting the techniques of \cite{VZ12} and \cite{HRV10}. We recall a number of technical lemmas from these papers, and observe that since these lemmas rely only on quantum friendly black-box techniques, they also hold in the QCCC setting. While we will focus on the non-uniform setting, the results should also hold against uniform adversaries by adapting the arguments of \cite{VZ12} and \cite{HRV10} in the uniform setting.

We first recall the following definitions used in~\cite{VZ12,HRV10} (all of which operate in the non-uniform setting)

\begin{definition}
    Let $X,B$ be two jointly sampled random variables. We say that $B$ has $(t,\epsilon)$ (quantum) pseudoentropy at least $k$ given $X$ if there exists a random variable $C$ jointly sampled with $X$ such that
    \begin{enumerate}
        \item $H(C|X)\geq k$
        \item $(X,B)$ and $(X,C)$ are $(t,\epsilon)$-indistinguishable (by quantum circuits)
    \end{enumerate}

    When $t = 1/\negl$ and $\epsilon = \negl$ we can omit the $t,\epsilon$ and simply say that $B$ has (quantum) pseudoentropy $k$.
\end{definition}

\begin{definition}
    Let $B^{(i)}$ be a random variable over $[q]$ for each $i\in [m]$. We say that $B = (B^{(1)},\dots, B^{(m)})$ has next-block (quantum) pseudoentropy at least $k$ if the random variable $B^{(I)}$ has (quantum) pseudoentropy at least $k/m$ given $B^{(1)},\dots,B^{(I-1)}$, for $I\sim [m]$. If $q = 2$, then we will use the term next-bit pseudoentropy.
\end{definition}

\begin{definition}
    Let $B^{(i)}$ be a random variable for each $i\in [m]$. We say that every block of $B = (B^{(1)},\dots, B^{(m)})$ has next-block (quantum) pseudoentropy at least $k$ if for all $i$, the random variable $B^{(i)}$ has (quantum) pseudoentropy at least $k$ given $B^{(1)},\dots,B^{(i-1)}$.
\end{definition}

\begin{definition}
    For random variables $A,B$ the KL divergence from $A$ to $B$ is defined as
    $$KL(A || B) = \E_{a \sim A}\left[\log \frac{\Pr[A\to a]}{\Pr[B \to a]}\right]$$
\end{definition}

\begin{definition}
    Let $X$, $B$ be jointly sampled random variables. We say that $B$ is $(t,\delta)$ (quantum) KL-hard for sampling given $X$ if for all size-$t$ randomized (quantum) circuits $S$,
    $$KL(X, B || X, S(X)) > \delta$$
\end{definition}

\begin{remark}
    We will also define pseudo-min-entropy to have the same definition as pseudoentropy, but with Shannon entropy replaced with min-entropy. We analogously define next-block pseudo-min-entropy.
\end{remark}

\subsubsection{$\OWP$ imply next-bit pseudoentropy}

We state quantum versions of the corresponding lemmas from \cite{VZ12}

\begin{lemma}[Quantum version of Theorem 3.15 and Lemma 3.6 from \cite{VZ12}]\label{lem:kltoentropy}
    Let $(X,B)$ be jointly sampled random variables over $\{0,1\}^n \times [q]$. If $B$ is quantum $(t,\delta)$ KL-hard for sampling given $X$, then for every $\epsilon>0$, $B$ has quantum $(t', \epsilon)$ psuedoentropy at least $H(B|X) + \delta - \epsilon$ given $X$, for $t' = t^{\Omega(1)}/\poly(n,q,1/\epsilon)$.
\end{lemma}

\begin{proof}
    The proof is exactly the same as the proofs of Lemma 3.6 and Theorem 3.15 from \cite{VZ12}, but with the use of the Min-Max theorem replaced by its quantum equivalent (Theorem 4.1 from \cite{chen2017computational}).
\end{proof}

\begin{lemma}[Quantum chain rule for KL-hardness (variant of Lemma 4.3 from \cite{VZ12}]\label{lem:chainrule}
    Let $Y$ be a distribution over $\{0,1\}^n$, jointly distributed with $Z$. If $Y$ is quantum $(t,\delta)$ KL-hard for sampling given $Z$, then $Y_I$ is quantum $(t', \delta/n)$ KL-hard for sampling given $(Z,Y_1,\dots,Y_{I-1})$, for $I \sim [n]$, $t'=t/O(n)$.
\end{lemma}

\begin{proof}
    Same as Lemma 4.3 from \cite{VZ12}.
\end{proof}

\begin{theorem}\label{thm:owptonbpe}
Let $(\Samp,\Ver)$ be a $(\omega,\gamma)$-$\OWP$ with puzzle length $m(\lambda)$ secure against all time $t$ adversaries. Then for all $\epsilon>0$, $\Samp \to (k, s)$ has quantum $(t',\epsilon/m)$ next-bit pseudoentropy at least $$H(k,s) + \delta - \epsilon$$ for $t' = t^{\Omega(1)}/\poly(\lambda,1/\epsilon)$ and
$$\delta = (1-\omega)\log\frac{1-\omega}{\gamma} + \omega \log \frac{\omega}{1-\gamma}$$
\end{theorem}

\begin{proof}
We first observe that $k$ is $(t,\delta)$ KL-hard to sample given $s$. Let $S$ be any time $t$ quantum circuit. We will show that $KL(s,k || s, S(s)) \geq \delta$.
    
    By monotonicity of quantum relative entropy,
    $$KL(\Ver(s,k) || \Ver(s, S(s))) \leq KL(s,k || s, S(s))$$
    But observe that $\Ver(s,k)$ is a Bernoulli random variable $Bern(p)$ for some $p > 1 - \omega$ by correctness. But $\Ver(s, S(s))$ is a Bernoulli random variable $Bern(p')$ for $p'$ the advantage of $S$ in the $\OWP$ game. By security of the $\OWP$, $p' \leq \gamma$.
    $$KL(\Ver(s,k) || \Ver(s,S(s))) \geq KL(Bern(1-\omega), Bern(\gamma)) $$
    $$= (1-\omega)\log\frac{1-\omega}{\gamma} + \omega \log \frac{\omega}{1-\gamma} = \delta.$$

    Let $m = \poly(\lambda)$ be the length of $k$. By \Cref{lem:chainrule}, for $I\sim [m]$, $k_I$ is quantum $\left(\frac{t}{O(m)}, \frac{\delta}{m}\right)$ KL-hard for sampling given $(s,k_1,\dots,k_{I-1})$.

    And so by \Cref{lem:kltoentropy}, for every $\epsilon>0$, $k_I$ has quantum $(t',\epsilon/m)$ pseudoentropy at least $H(k_I|s,k_1,\dots,k_{I-1}) + \frac{\delta}{m} - \epsilon/m$ given $s$. But by definition, this means that $k$ has $(t',\epsilon/m)$ next-bit pseudoentropy $m\cdot H(k_I|s,k_1,\dots,k_{I-1}) + \delta - \epsilon$ given $s$. But $m\cdot H(k_I | s,k_1,\dots,k_{I-1}) \geq H(k|s)$, so $(k,s)$ has $(t',\epsilon/m)$ next-bit pseudoentropy $H(s) + H(k|s) + \delta - \epsilon = H(k,s) + \delta - \epsilon$ for $t' = t^{\Omega(1)}/\poly(\lambda,1/\epsilon)$.
\end{proof}

\subsubsection{Next-bit pseudoentropy implies non-uniform EFID}

This argument follows the framework of \cite{HRV10}. The argument is nearly identical to the classical case. The primary difference is that instead of lower bounding the length of the resultant string by the length of the input, we instead lower bound the length of the resultant string by the entropy of the input.

\begin{lemma}[Lemma 5.2 from~\cite{HRV10}]\label{lem:hrv52}
    For $i \in [m]$, $x^{(1)},\dots,x^{(\ell)} \in \mathcal{M}^m$, we define 
    $$Equalizer(i, x^{(1)},\dots, x^{(\ell)})\coloneqq x^{(1)}_{i},\dots, x^{(1)}_{m},x^{(2)},\dots, x^{(\ell - 1)},x^{(\ell)}_1,\dots, x^{(\ell)}_{i-1}$$
    That is, $Equalizer(i, x^{(1)},\dots, x^{(\ell)})$ truncates the first $i-1$ blocks from the first input and the last $m - (i - 1)$ blocks from the last input.
    
    Let $X$ be a random variable over $\mathcal{M}^m$ with $(t,\epsilon)$ next block quantum pseudoentropy at least $k$. Let $X^{(1)},\dots,X^{(\ell)}$ be $\ell$ independent and identically distributed copies of $X$, and let $I$ be uniformly distributed over $[m]$. Define $\wt{X} = Equalizer(I,X^{(1)},\dots,X^{(\ell)})$. Then every block of $\wt{X}$ has $(t - O(\ell\cdot m\cdot \log \abs{\mathcal{M}}), \ell \cdot \epsilon)$ next-block quantum pseudoentropy at least $k/m$.
\end{lemma}

\begin{lemma}[Lemma 5.3 from~\cite{HRV10}]\label{lem:hrv53}
    Let $X$ be a random variable over $\mathcal{M}^m$ where every block of $X$ has $(t,\epsilon)$ next-block quantum pseudo-min-entropy at least $k$. Let $X^a$ refer to $a$ i.i.d. copies of $X$. For every $\kappa > 0$, $X^a$ has $(t',\epsilon')$ next-block quantum pseudo-min-entropy $k'$ where
    \begin{enumerate}
        \item $t' = t - O(ma\log \abs{\mathcal{M}})$
        \item $\epsilon' = a^2 (\epsilon + 2^{-\kappa} + 2^{-ca})$ for a universal constant $c>0$
        \item $k' = ak - O(\log(a\abs{\mathcal{M}})\sqrt{a\kappa})$
    \end{enumerate}
\end{lemma}

\begin{lemma}[Lemma 5.4 from~\cite{HRV10}]\label{lem:hrv54}
    There exists an efficient procedure $\Ext \in NC^1$ on input $x\in(\{0,1\}^a)^m$ and $s \in \{0,1\}^a$ which outputs a string of length $y \in \{0,1\}^{a + m(k-k')}$ such that the following holds:
    Let $X$ be a random variable over $(\{0,1\}^a)^m$ such that every block of $X$ has $(t,\epsilon)$ next-bit quantum pseudo-min-entropy $k$, then for all QPT $\A$ running in time $t- m\cdot a^{O(1)}$
    $$\abs{\Pr[\A(\Ext(X,\U_a))\to 1] - \Pr[\A(\U_{a + m(k-k')}) \to 1]} \leq m(\epsilon + 2^{-k'/2})$$
\end{lemma}

\begin{proof}
    The reductions for all three of these lemmas from~\cite{HRV10} are fully black-box (in a quantum friendly way), and also hold in the quantum setting.
\end{proof}

\begin{lemma}\label{lem:entropytosd}
    Let $X$ be a random variable with $|X| = m$. If $H(X) \leq m - \delta$ for some $\delta > 0$, then
    $$SD(X,\U_{m}) \geq \frac{\delta}{2m-\delta} - 2^{-\delta/2}$$

    In particular, if $m = O(p(\lambda))$ and $\delta = \Omega(p'(\lambda))$ for some polynomials $p$ and $p'$, then there exists a polynomial $q$ such that 
    $$SD(X, \U_m) =\Omega\left(\frac{1}{q(\lambda)}\right)$$
\end{lemma}

\begin{proof}[Proof of \cref{lem:entropytosd}]
    $$S \coloneqq \{ x : \Pr[X \to x] > 2^{-m + \delta/2}\}$$

    Observe that $1 \geq \sum_{x \in S} \Pr[X \to x] \geq |S| 2^{-m + \delta/2}$ and so $|S| \leq 2^{m - \delta/2}$. Thus,
    $$\Pr[\U \in S] \leq 2^{m - \delta/2  -m} = 2^{-\delta/2}$$
    
    We have
    $$H(X) = \E_{X \to x}[-\log \Pr[X\to x]] \leq m - \delta$$
    By Markov bound, we get
    $$\Pr_{X \to x}[-\log \Pr[X \to x] \geq m - \delta/2] \leq \frac{m - \delta}{m - \delta/2}$$
    and so
    $$\Pr[X \in S] = \Pr_{X \to x}[\Pr[X \to x] > 2^{-m + \delta/2}] \geq 1 - \frac{m - \delta}{m - \delta/2} = \frac{\delta}{2m - \delta}$$

    Putting these two statements together, we have
    $$SD(X,\U) \geq \Pr[X \in S] - \Pr[\U \in S] \geq \frac{\delta}{2m-\delta} - 2^{-\delta/2}$$
\end{proof}

\begin{theorem}[Adapted from Theorem 5.5 from~\cite{HRV10}]
    Let $m(\lambda),\Delta(\lambda)$ be two computable functions such that $\Delta =\Delta(\lambda) \in [1/\poly(\lambda), \lambda]$. Let $\{X\}_{\lambda\in\N}$ be a family of efficiently samplable random variables of length $m(\lambda)$ with next-bit pseudoentropy at least $H(X) + \Delta$. Then there exists a function $\nu^*(\lambda) \leq |X|$ such that there exists a QPT algorithm $D_\nu(1^\lambda)$ outputting a classical string such that
    \begin{enumerate}
        \item $D_{\nu^*}(1^\lambda) \approx \U$
        \item $SD(D_{\nu^*}(1^\lambda),\U) \geq \frac{1}{p(\lambda)}$
    \end{enumerate}
    for some efficiently computable polynomial $p$.
\end{theorem}

The proof of this theorem follows essentially the same lines as the proof from~\cite{HRV10}, but with all references to the input length replaced by the entropy of $X$. This gives a distribution indistinguishable from uniform but with less than full entropy.~\Cref{lem:entropytosd} then gives a bound on the statistical distance. However, this bound is not $1-\negl$. Fortunately, taking the product distribution amplifies statistical distance, so we simply take the direct product of the construction from~\cite{HRV10}.

\begin{proof}

    Without loss of generality, we will assume that the block-length of $X$ is a power of $2$ (otherwise, we can just append $0$s). We have that for all $c > 0$ and sufficiently large $\lambda$, $X$ has $(t=\lambda^c,\epsilon = \lambda^{-c})$ next-bit pseudoentropy $H(X)+\Delta$.  We set $\ell \coloneqq \ceil{2(\nu^* + \Delta + \log m)/\Delta} = \Omega(\nu^*/\Delta)$ in~\Cref{lem:hrv52} and get a new random variable $\wt{X} = Equalizer(I,X^{(1)},\dots, X^{(\ell)})$. We will define $$D_\nu(1^\lambda, 1) \coloneqq \Ext((\wt{X})^a, \U_a)$$
    using the $\Ext$ from~\Cref{lem:hrv54} for some value of $a = \poly(\lambda)$ and with output length $d_\nu \coloneqq a + m(\ell - 1)(k'_\nu - \kappa)$ for $k'_\nu \coloneqq ak_\nu - O(\log(a|\mathcal{M}|)\sqrt{a\kappa})$ and $k_\nu \coloneqq (\nu+\Delta)/m$.

    Observe that when $\nu = \nu^* = H(X)$,~\Cref{lem:hrv52} shows that every bit of $\wt{X}$ has $(t - O(\ell m), \ell \epsilon)$ next-bit quantum pseudoentropy at least  $k_{\nu^*} = (\nu^*+\Delta)/m$.

    Next,~\Cref{lem:hrv53} shows that every block of $(\wt{X})^a$ has $(t - O(m\ell a), a^2(\epsilon + 2^{-\kappa} +2^{-\Omega(a)}))$ next-bit quantum pseudo-min-entropy at least $k'_{\nu^*} = ak_{\nu^*} - O(\sqrt{a\kappa} \log a)$.

    Finally,~\Cref{lem:hrv54} shows that for all QPT $\A$ running in time $t-O(m\ell a^{O(1)}) = t - \poly(\lambda)$, 
    $$\abs{\Pr[\A(D_{\nu^*}(1^\lambda)\to 1] - \Pr[\A(\U)\to 1]} \leq m\ell (a^2(\epsilon + 2^{-\kappa} +2^{-\Omega(a)}) + 2^{-\kappa/2}) $$ 
    $$\leq \poly(\lambda)(\epsilon + 2^{-\kappa/2} + 2^{-\Omega(a)})$$

    Setting $\kappa = \lambda/2$ and using the fact that this holds for all $t = n^c,\epsilon = n^{-c}$, we get that $D_{\nu^*}(1^\lambda)$ and $\U_{d_\nu}$ are indistinguishable.

    It just remains to be shown that 
    $$\Delta(D_{\nu^*}(1^\lambda), \U_{d_{\nu^*}}) \geq \frac{1}{\poly(\lambda)}$$
    We will do this by showing that $H(D_{\nu^*}) \leq H(\U_{d_{\nu^*}}) - \Omega(\poly(\lambda)) = d_{\nu^*} - \Omega(\poly(\lambda))$ and then applying~\Cref{lem:entropytosd}. 

    Observe that when $\nu^* = H(X)$, $H(\wt{X}) \leq \ell H(X) + \log m$. And so $H(\wt{X}^{a}) \leq a (\ell H(X) + \log m)$. It is clear to see that $H(D_{\nu^*}) \leq a(\ell H(X) + \log m) + a$. Let us denote this value by $d' \coloneqq H(D_{\nu^*})$.

    \begin{equation*}
        \begin{split}
            d_{\nu^*} = a + m(\ell - 1)(k'_{\nu^*} - \kappa)\\
            = a + m(\ell - 1)(ak_{\nu^*} - O(\sqrt{a\kappa} \log a) - \kappa)\\
            = a + a(\ell H(X) + \log m) + a \ell \Delta - a(\nu^* + \Delta + \log m) - m(\ell - 1)(O(\sqrt{a\kappa} \log a) + \kappa)\\
            \geq a + a(\ell H(X) + \log m) + a \ell \Delta / 2 - m(\ell - 1)(O(\sqrt{a\kappa} \log a) + \kappa)\\
            \geq d' + a \ell \Delta / 2 - m(\ell - 1)(O(\sqrt{a\kappa} \log a) + \kappa)\\
            \geq d' + a \ell \Delta / 4\\
            = d' + \Omega(a \nu^*)\\
            = d' + \Omega(\poly(\lambda))
        \end{split}
    \end{equation*}

    when
    $$a = \Theta\left(\left(\left(\frac{m(\ell - 1)}{\Delta \ell}\right)^2\kappa \log^2\left(\frac{m(\ell - 1)\kappa}{\Delta \ell}\right)\right)\right) = \Theta\left(\frac{m^2 \kappa \log^2 \lambda}{\Delta^2}\right)$$
\end{proof}

\begin{lemma}[Amplification of statistical distance.]\label{lem:sdamp}
    Let $SD(X,Y) \geq \delta$. Then if $q \geq \frac{12t}{\delta^2}$ $SD(X^q,Y^q) \geq 1 - 2e^{-t}$.
\end{lemma}

\begin{proof}[Proof of \cref{lem:sdamp}]
    Since $SD(X,Y) \geq \delta$, we know there exists a set $S$ such that
    $$\Pr[X\in S] - \Pr[Y \in S] \geq \delta$$
    Define $\alpha \coloneqq \Pr[X \in S]$ and $\beta \coloneqq \Pr[Y \in S]$. We will define a new set
    $$S' \coloneqq \{(x_1,\dots,x_q) : \text{ at least an }\frac{\alpha + \beta}{2}\text{ fraction of }x_i \in S\}.$$
    By the Chernoff bound,
    \begin{equation*}
        \begin{split}
            \Pr[X^q \notin S'] \leq e^{-\left(1 - \frac{\alpha+\beta}{2\alpha}\right)^2 \alpha q/2}\\
            \leq e^{-q\frac{(\alpha-\beta)^2}{8\alpha}}\\
            \leq e^{-q\frac{\delta^2}{8}}
        \end{split}
    \end{equation*}
    Similarly,
    \begin{equation*}
    \begin{split}
        \Pr[Y^q \in S] \leq  e^{-\left(\frac{\alpha - \beta}{2\beta}\right)^2 \beta q/3}\\
        = e^{-q\frac{(\alpha-\beta)^2}{12 \beta}}\\
        \leq e^{-q\frac{\delta^2}{12}}
    \end{split}
    \end{equation*}

    So if $q \geq \frac{12t}{\delta^2}$,
    $$SD(X^q,Y^q) \geq 1 - 2e^{-q\frac{\delta^2}{12}} = 1 - 2e^{-t}$$
\end{proof}

\begin{corollary}\label{cor:nbpetoqefid}
    Let $\Delta =\Delta(n) \in [1/\poly(n), n]$ and let $\{X\}_{n\in\N}$ be a family of efficiently samplable random variables of length $m$ with next-bit pseudoentropy at least $H(X) + \Delta$. Then there exists a function $\nu^*(\lambda) \leq |X|$ such that there exists a non-uniform $\nu^*$-EFID.
\end{corollary}

\begin{proof}
    Define $G_\nu(1^\lambda,0)$ to be uniform over $q(\lambda) \cdot d_s$ bits for some $q=\poly(\lambda)$ to be set later. Define $G_\nu(1^\lambda,1) = D_{\nu^*} (1^\lambda)^{q(\lambda)}$.

    Since $D_{\nu^*}(1^\lambda) \approx \U_{d_\nu}$, $D_{\nu^*}(1^\lambda)^{q(\lambda)} \approx \U_{q(\lambda)\cdot d_\nu}$.

    Since $SD(D_{\nu^*}(1^\lambda),\U_{d_\nu}) \geq \frac{1}{p(\lambda)}$, if we define $q(\lambda) = 12\lambda p(\lambda)$, then by~\Cref{lem:sdamp}
    $$SD(G_\nu(1^\lambda,0), G_\nu(1^\lambda, 1)) \geq 1 - 2e^{-\lambda} = 1 -\negl(\lambda)$$
\end{proof}

\begin{corollary}\label{cor:owptoqefidfinal}
    If, for some $c>0$, there exists a $(\negl(\lambda), 1-\lambda^{-c})$ one-way puzzle $(\Samp,\Ver)$, then there exists a $\nu^*$-non-uniform EFID pair with $\nu^*\leq m + n$.
\end{corollary}

\begin{proof}
    Let $\omega = \negl(\lambda)$ and $\gamma = 1-\lambda^{-c}$ be the correctness and security parameters of $(\Samp,\Ver)$, and let $n(\lambda)$ be the length of the key and let $m(\lambda)$ be the length of the puzzle. For all $t = \poly(\lambda)$, $\epsilon>0$, by~\Cref{thm:owptonbpe}, $\Samp \to (k,s)$ has $(t^{\Omega(1)}/\poly(\lambda,1/\epsilon),\epsilon/m)$ quantum next-bit pseudoentropy at least $H(k,s) + \delta - \epsilon$ for $$\delta = (1-\omega)\log\frac{1-\omega}{\gamma} + \omega\log\frac{\omega}{1-\gamma}.$$ But observe,
    \begin{equation*}
        \begin{split}
            (1-\omega)\log\frac{1-\omega}{\gamma} + \omega\log\frac{\omega}{1-\gamma}\\
            \geq \frac{1}{2}\log \frac{1}{1-\lambda^{-c}} + (1-\omega)\log (1-\omega) + \omega \log \omega + \omega\log \lambda^c\\
            \geq \frac{1}{2}\log \frac{1}{1-\lambda^{-c}} - \negl(\lambda)
            \geq \lambda^{-(c+1)}
        \end{split}
    \end{equation*}
    for all sufficiently large $\lambda$.

    Thus, for all sufficiently large $d$ such that $n^{-d}\cdot m \leq \lambda^{-(c+1)}/2$, for all $t = \poly(\lambda)$, $\Samp \to (k,s)$ has $(t^{\Omega(1)}/\poly(\lambda,\lambda^d), n^{-d})$ quantum next-bit pseudoentropy at least $H(k,s) + \frac{1}{2}\lambda^{-(c+1)}$. Adjusting the value of $t$ shows us that $\Samp\to (k,s)$ has quantum next-bit pseudoentropy at least $H(k,s) + \frac{1}{2}\lambda^{-(c+1)}$.

    By~\Cref{cor:nbpetoqefid} there exists a $\nu^*$-EFID for $\nu^*\leq m + n$.
\end{proof}

\begin{corollary}\label{cor:weakowptoefi}
    If, for some $c>0$, there exists a $(\negl(\lambda), 1-\lambda^{-c})$ one-way puzzle $(\Samp,\Ver)$, then there exists an EFI pair.
\end{corollary}

\begin{proof}
    Prior work shows that there exists a quantum combiner for quantum bit commitments~\cite{HKNY23}, which are equivalent to EFI pairs~\cite{BCQ22}. We observe that this combiner (when composed with the construction of EFI pairs from commitments) operates separately on each security parameter. Thus, given a non-uniform EFID pair, we can apply the combiner from~\cite{HKNY23} to the non-uniform construction instantiated with each possible value of the advice. This process maintains security, but produces a quantum output. Thus, this process takes a non-uniform EFID pair and produces an EFI pair.

    Since~\Cref{cor:owptoqefidfinal} shows that weak one-way puzzles can be used to build a non-uniform EFID pair, composing that construction with this approach produces an EFI pair from any weak one-way puzzle.
\end{proof}


\subsection{QEFID imply $\OWP$}
\begin{definition}[$\nu^*$-non-uniform $\OWP$]
    Let $\nu^*(\lambda)$ be some function. A $\nu^*$-non-uniform one way puzzle ($\OWP$) is a pair of a sampling algorithm and a verification function $(\Samp_\nu,\Ver_\nu)$ taking in advice $\nu$ with the following syntax:
    \begin{enumerate}
        \item $\Samp_\nu(1^{\lambda}) \to (k,s)$ is a uniform QPT algorithm which outputs a pair of classical strings $(k,s)$. We refer to $s$ as the puzzle and $k$ as the key. Without loss of generality, we can assume $k \in \{0,1\}^\lambda$.
        \item $\Ver_\nu(k,s) \to b$ is a function which takes in a key and puzzle and outputs a bit $b \in \{0,1\}$.
    \end{enumerate}
    satisfying the following properties:\\
    \begin{enumerate}
        \item Correctness: Outputs of the sampler pass verification with overwhelming probability
        $$\Pr_{\Samp_{\nu^*}(1^{\lambda})\to (k,s)}[\Ver_{\nu^*}(k,s)\to 1] \geq 1 - \alpha$$
        \item Security: Given a puzzle $s$, it is computationally infeasible to find a key $s$ which verifies. That is, for all non-uniform QPT algorithms $\A$,
        $$\Pr_{\Samp_{\nu^*}(1^{\lambda}) \to (k,s)}[\Ver_{\nu^*}(\A(s),s) \to 1] \leq \beta$$
    \end{enumerate}

    That is, a non-uniform one-way puzzle is a one-way puzzle for which correctness and security are only guaranteed to hold when given the correct advice.
\end{definition}

\begin{lemma}[From QEFID to $\OWP$]\label{lem:qefid-to-owp}
    If there exists QEFID pair $G$, then there exists a $\OWP$ $(\Samp,\Ver)$.
\end{lemma}

\begin{proof}[Proof of \cref{lem:qefid-to-owp}]
    We will define the $\OWP$ as follows
    \begin{enumerate}
        \item $\Samp(1^\lambda)$: Sample $k$ uniformly at random from $\{0,1\}^\lambda$. For each $i\in[\lambda]$ sample $s_i$ from $G(1^\lambda, k_i)$. Output $(k,(s_1,\dots,s_\lambda))$.
        \item $\Ver(1^\lambda,k,s)$: Set $$T^* = \argmax_{T:\{0,1\}^m\to \{0,1\}}\left(\Pr[T(G(1^\lambda, 1)) \to 1] - \Pr[T(G(1^\lambda, 0)) \to 1\right)$$
        For each $i \in [\lambda]$, check if $T^*(s_i) = k_i$. If all tests pass, output $1$. Otherwise, output $0$.
    \end{enumerate}

    To prove correctness, let us observe that since $SD(G(1^\lambda,0), G(1^\lambda, 1)) \geq 1-\negl(\lambda)$, there exists a $T^*$ such that $$\left(\Pr[T^*(G(1^\lambda, 1)) \to 1] - \Pr[T^*(G(1^\lambda, 0)) \to 1\right) \geq 1-\negl(\lambda).$$
    In particular, for any such $T^*$, we have $\Pr[T^*(G(1^\lambda,1)\to 1] \geq 1-\negl(\lambda)$.
    Thus, $\Pr[\Ver(1^\lambda,\Samp(1^\lambda))\to 1] \geq (1-\negl(\lambda))^\lambda \geq 1-\negl(\lambda)$.

    To prove security, we observe that the QEFID game is a three-round quantum interactive protocol. Thus, we can apply quantum amplification (Theorem 4.1 from~\cite{bostanci2023efficient}) to see that for all non-uniform QPT $\A$,
    \begin{equation*}
        \begin{split}
            \Pr_{\Samp(1^\lambda)\to(k,(s_1,\dots,s_\lambda))}[\A(s_1,\dots,s_\lambda) = k]\\
            \leq \left(\Pr_{\{0,1\}\to b,G(1^\lambda,b)\to s}[\A(s) \to b]\right)^\lambda\\
            \leq \left(\frac{1}{2} + \negl(\lambda)\right)^\lambda\\
            \leq \negl(\lambda).
        \end{split}
    \end{equation*}
    To conclude, we observe that verification simply checks whether $k$ is equal to the output of $T^*$ on all of $(s_1,\dots,s_\lambda)$. Thus, the only way to invert the one-way puzzle is to output the key used for generation. Formally,
    \begin{equation*}
        \begin{split}
            \Pr_{\Samp(1^\lambda)\to(k,(s_1,\dots,s_\lambda))}[\Ver(\A(s_1,\dots,s_\lambda),(s_1,\dots,s_\lambda))]\\
            \leq \Pr_{\Samp(1^\lambda)\to(k,(s_1,\dots,s_\lambda))}[\A(s_1,\dots,s_\lambda) = k]\\ + \Pr_{\Samp \to (k,(s_1,\dots,s_\lambda))}[\Ver(k',s) \to 1 | \A(s_1,\dots,s_\lambda) \to k' \neq k]\\
            \leq \negl(\lambda) + \Pr_{\Samp \to (k,(s_1,\dots,s_\lambda))}[\text{ there exists an index }i\text{ such that }T^*(s_i) \neq k_i]\\
            \leq \negl(\lambda) + \lambda \Pr_{\{0,1\}\to b,G(b) \to s}[T^*(s) \neq b]\\
            \leq \negl(\lambda) + \lambda \negl(\lambda)
            =\negl(\lambda)
        \end{split}
    \end{equation*}
\end{proof}

We observe that the same argument also works relative to an advice string, and so we have

\begin{lemma}[From non-uniform QEFID to non-uniform $\OWP$]
    Let $\nu^*(\lambda)$ be some function. If there exists a $\nu^*$-non-uniform QEFID pair $G_\nu$, then there exists a $\nu^*$-non-uniform $\OWP$ $(\Samp_\nu,\Ver_\nu)$.
\end{lemma}

\begin{lemma}[From non-uniform $\OWP$ to $\OWP$]
     Let $p(\lambda)$ be some computable polynomial. Let $\nu^*(\lambda)$ be some function satisfying $\nu^*(\lambda)\leq p(\lambda)$. If there exists a $\nu^*$-non-uniform $\OWP$ $(\Samp_\nu,\Ver_\nu)$, then there exists a $\OWP$ $(\wt{\Samp},\wt{\Ver})$.
\end{lemma}

\begin{proof}
    We simply apply the construction from~\Cref{cor:owpcombmany} to $(\Samp_1,\Ver_1),\dots,$ $(\Samp_p,\Ver_p)$. An analogous argument to the proof that this construction is a combiner will give that the resulting $(\wt{\Samp},\wt{\Ver})$ is a $\OWP$.
\end{proof}

We now have all the pieces to show \cref{thm:owpamplification} that security for $\OWP$ can be amplified. First ~\Cref{cor:owptoqefidfinal} gives us that weak $\OWP$ imply non-uniform QEFID, Then the two lemmas above give us that non-uniform QEFID imply non-uniform $\OWP$ which in turn imply strong $\OWP$.


\section{Equivalence of $\OWP$ to variants}

\subsection{Random Input $\OWP \leftrightarrow \OWP$}

In this section we answer an open question left by \cite{khurana2024commitments} of whether one-way puzzles imply random input one-way puzzles.

\begin{definition}
    A random input one-way puzzle is a pair of a sampling algorithm and a verification function $(\PuzzSamp,\Ver)$ with the following syntax:
    \begin{enumerate}
        \item $\PuzzSamp(1^\lambda,k) \to s$ is a uniform QPT algorithm which takes in a key $k$ and outputs a puzzle $s$.
        \item $\Ver(1^\lambda, k,s) \to b$ is a function which takes in a key and a puzzle and outputs a bit $b \in \{0,1\}$
    \end{enumerate}
    satisfying the following properties:
    \begin{enumerate}
        \item Correctness: Outputs of the sampler pass verification with overwhelming probability
        $$\Pr_{\{0,1\}^\lambda \to k,\PuzzSamp(1^\lambda,k)\to s}[\Ver(k,s) \to 1] \geq 1 - \negl(\lambda)$$
        \item Security: Given a puzzle $s$, it is computationally infeasible to find a key $s$ which verifies. That is, for all non-uniform QPT algorithms $\A$,
        $$\Pr_{\{0,1\}^\lambda \to k,\PuzzSamp(1^\lambda,k)\to s}[\Ver(\A(s),s)\to 1] \leq \negl(\lambda)$$
    \end{enumerate}
\end{definition}

\begin{theorem}\label{thm:randominput}
    There exists a one-way puzzle $(\Samp,\Ver)$ if and only if there exists a random input one-way puzzle $(\PuzzSamp',\Ver')$. If $\Ver$ is efficient, then so is $\Ver'$.
\end{theorem}

\begin{proof}
    Any random input one-way puzzle gives a one-way puzzle by having the sampler just sample the random key itself. Thus, we focus on building a random input one-way puzzle from any one-way puzzle.

    We will define the random input one-way puzzle as follows
    \begin{enumerate}
        \item $\PuzzSamp'(1^\lambda, k')$: Run $\Samp \to (k,s)$. Output $s' = (k \xor k', s)$.
        \item $\Ver'(k',s')$: Parse $s'$ as $(a,b)$. Output $\Ver(a\xor k', b)$.
    \end{enumerate}

    Observe that $\Ver(k',s') = \Ver(k' \xor (k \xor k'), s) = \Ver(k,s)$ and so correctness follows from correctness of $(\Samp,\Ver)$.

    We will show security using a reduction. Let $\A$ be any adversary such that
    $$\Pr_{\{0,1\}^\lambda \to k',\PuzzSamp(k') \to s'}[\Ver'(\A(s'),s') \to 1] \geq \epsilon$$

    We will define $\A'$ as follows. On input $s$, sample $r$ uniformly at random. Output $\A(r,s) \xor r$.

    Note that $(r,s)$ is identically distributed to the output distribution of $\PuzzSamp$ on random inputs. Thus, 
    \begin{equation*}
        \begin{split}
            \Pr_{\Samp \to (k,s)}[\Ver(\A'(s),s) \to 1]\\
            =\Pr_{\{0,1\}^\lambda \to k',\PuzzSamp(k') \to (r,s)}[\Ver(\A(r,s) \xor r, s)]\\
            =\Pr_{\{0,1\}^\lambda \to k',\PuzzSamp(k') \to s'}[\Ver'(\A(s'),s') \to 1] \geq \epsilon
        \end{split}
    \end{equation*}
    and so $\epsilon \leq \negl(\lambda)$.

    Note that if $\Ver$ was originally efficient, so is $\Ver'$, and so this same construction works for efficiently verifiable one-way puzzles, producing a random input efficiently verifiable one-way puzzle.
\end{proof}

The philosophical message behind this theorem is that it doesn't matter whether or not the key and puzzle are sampled together, the fundamental difference between one-way puzzles and one-way functions is that the puzzle is sampled using a quantum algorithm instead of classical randomness.

\subsection{Distributional $\OWP \leftrightarrow \OWP$}

Here we show how to build $\OWP$ from distributional $\OWP$.

\begin{definition}
    A $\beta$ distributional one-way puzzle is a uniform QPT algorithm $\Samp(1^\lambda) \to (k,s)$ which takes in a security parameter $\lambda$ and produces a key $k$ and a puzzle $s$ such that given a puzzle $s$, it is computationally infeasible to sample from the conditional distribution over keys. More formally, we require that for all non-uniform QPT algorithms $\mathcal{A}$, for all sufficiently large $\lambda$,
    $$\Samp(1^\lambda) \to (k,s)$$
    $$\Delta((k,s),(\mathcal{A},s)) \geq \beta(\lambda)$$
\end{definition}

\begin{theorem}
    If, for some $c>0$, there exists a $\lambda^{-c}$ distributional one-way puzzle $\Samp$, then there exists a strong one-way puzzle.
\end{theorem}

This theorem follows directly from Pinsker's inequality, which states
\begin{theorem}[Pinsker's inequality]
    Let $P$ and $Q$ be any two probability distributions. Then
    $$\Delta(P,Q) \leq \sqrt{\frac{\ln 2}{2}KL(P||Q)}.$$
\end{theorem}

Thus, if a puzzle is $\lambda^{-c}$ distributionally one-way, then we have that for all QPT $S$,
$$KL(s,k || s, S(s)) \geq \frac{12}{\ln 2} \Delta(s,k || s, S(k))^2 \geq \lambda^{-2c}$$
Reading the proof of~\Cref{thm:owpamplification}, it is clear that the only property required of the one-way puzzle is that 
$$KL(s,k || s, S(s)) \geq \frac{1}{\poly(\lambda)},$$
and so applying the construction of~\Cref{thm:owpamplification} to $\Samp$ gives a strong one-way puzzle.
\section{Acknowledgments}
We thank Yanyi Liu for insightful discussion. We thank Bruno Cavalar for discussions concerning distributional one way puzzles. Kai-Min Chung was partially supported by the Air Force Office of Scientific Research under award number FA2386-23-1-4107 and NSTC QC project, under Grant no. NSTC 112-2119-M-001-006. E. Goldin was supported by a National Science Foundation Graduate Research Fellowship.

\bibliographystyle{alpha}	
\bibliography{refs}

\end{document}